\documentclass[11pt,twoside,reqno]{amsart}

\usepackage{graphicx, caption}
\usepackage{subfig}

\usepackage[mathscr]{eucal}

\usepackage{amsbsy,amsfonts,color}
\usepackage{amstext}

\usepackage{hyperref}

\usepackage{latexsym,eucal,amsmath,amsthm, mathtools}
\usepackage{amssymb}

\newcommand{\blob}{\rule[.2ex]{.8ex}{.8ex}}
\newcommand{\dd}{\rule[.2ex]{.4ex}{.4ex}}

\newcommand{\field}[1]{\mathbb{#1}}
\newtheorem{definition}{Definition}[section]
\newtheorem{theorem}{Theorem}[section]
\newtheorem{lemma}{Lemma}[section]

\newtheorem{proposition}{Proposition}[section]
\newtheorem{remark}{Remark}[section]
\numberwithin{equation}{section}

\newcommand{\sabs}[1]{\left| #1 \right|} 
\newcommand{\abs}[1]{\bigl| #1 \bigr|} 
\newcommand{\norm}[1]{\lVert#1\rVert} 
\newcommand{\bnorm}[1]{\Bigl\| #1\Bigr\|} 
\newcommand{\avg}[1]{\left< #1 \right>} 

\newcommand{\x}{\underline{x}}
\newcommand{\y}{\underline{y}}
\newcommand{\z}{\underline{z}}
\newcommand{\m}{\underline{m}}
\newcommand{\li}{\underline{l}}
\newcommand{\alfa}{\underline{\alpha}}
\newcommand{\la}{\lambda}
\newcommand{\ka}{\kappa}
\newcommand{\kam}{\ka^{-1}}
\newcommand{\kamt}{\ka^{-2}}
\newcommand{\ep}{\epsilon}
\newcommand{\epm}{\ep^{-1}}
\newcommand{\epmt}{\ep^{-2}}
\newcommand{\Am}{A^{-1}}

\newcommand{\dx}{\partial_{x_1}}
\newcommand{\dy}{\partial_{x_2}}

\newcommand{\ro}{\underline{\rho}}

\newcommand{\om}{\omega}
\newcommand{\omm}{\underline{\omega}}
\newcommand{\T}{\mathbb{T}}
\newcommand{\N}{\mathbb{N}}
\newcommand{\Z}{\mathbb{Z}}
\newcommand{\R}{\mathbb{R}}

\newcommand{\C}{\mathbb{C}}

\newcommand{\A}{\mathcal{A}}
\newcommand{\Aa}{\underline{\A}}

\newcommand{\Rr}{\underline{\mathcal{R}}}
\newcommand{\Ss}{\underline{\mathcal{S}}}

\newcommand{\B}{\mathcal{B}}
\newcommand{\Bb}{\underline{\B}}

\newcommand{\Bad}{\mathscr{B}}

\newcommand{\shift}{{\rm T}}
\newcommand{\skews}{{\rm S}_\om}
\newcommand{\mshift}{{\rm T}_{\omm}}

\newcommand{\xx}{(x_1, x_2)}
\newcommand{\ab}{\underline{a}}
\newcommand{\aaa}{(a_1, a_2)}
\newcommand{\zz}{(z_1, z_2)}
\newcommand{\mm}{(m_1, m_2)}
\newcommand{\lli}{(l_1, l_2)}

\newcommand{\rover}{\frac{\rho}{2}}
\newcommand{\roverr}{\frac{\rho}{4}}

\newcommand{\less}{\lesssim}
\newcommand{\more}{\gtrsim}

\begin{document}

\title[Schr\"{o}dinger operators with multivariable Gevrey potentials] {Localization for quasiperiodic Schr\"{o}dinger operators with multivariable Gevrey potential functions}
\author{Silvius Klein}
\address{CMAF\\ Faculdade de Ci\^encias\\
Universidade de Lisboa\\
Portugal\\ 
and IMAR, Bucharest, Romania }
\email{silviusaklein@gmail.com}

\begin{abstract}{We consider an integer lattice quasiperiodic Schr\"{o}dinger operator. The underlying dynamics is either the skew-shift or the multi-frequency shift by a Diophantine frequency.  We assume that the potential function belongs to a Gevrey class on the multi-dimensional torus. Moreover, we assume that the potential function satisfies a generic transversality condition, which we show to imply a {\L}ojasiewicz type inequality for smooth functions of several variables.  Under these assumptions and for large coupling constant, we prove that the associated Lyapunov exponent is positive for all energies, and continuous as a function of energy, with a certain modulus of continuity.
Moreover, in the large coupling constant regime and for an asymptotically large frequency - phase set, we prove that the operator satisfies Anderson localization. 
}
\end{abstract}

\maketitle

\section{Definitions, notations, statement of main results}\label{introduction}
In this paper we study the one-dimensional lattice quasiperiodic Schr\"odinger operator $H (\x)$ acting on $l^2(\Z)$  by:
\begin{equation}\label{op1} 
[H (\x) \, \psi]_n := - \psi_{n+1} - \psi_{n-1} + \la \, v (\shift^n \x) \, \psi_n
\end{equation} 
where in equation (\ref{op1}):

$\blob$ $\x = \xx \in  \T^2$ is a parameter that introduces some randomness into the system;

$\blob$ $\la$ is a real number called the disorder of the system;
 
$\blob$  $v (\x) $ is a real valued function on $ \mathbb{T}^2 = (\mathbb{R} / \mathbb{Z})^2 $, that is, a real valued 
$ 1$-periodic function in each variable; 

$\blob$  $\shift$ is a specific ergodic transformation on $\T^2$, and $\shift^n$ is its $n$th iteration.

\smallskip

Some of the questions of interest regarding this, or other related operators, are the spectral types (pure point, absolutely continuous, singularly continuous), the topological structure of the spectrum, the rate of decay of the eigenfunctions, the positivity and regularity of the Lyapunov exponent, the regularity of the integrated density of states. 

Due to the ergodicity of the transformation $\shift$, the spectrum and the spectral types of the Hamiltonian system $ [ H (\x) ]_{\x \in \T^2}$ defined by (\ref{op1}) are not random - that is, they are independent of $\x$ almost surely (see \cite{CFKS}). 

A stronger property than pure point spectrum is Anderson localization, which for the physical model indicates an insulating behavior, while a purely absolutely continuous spectrum indicates metallic (conductive) behavior.

Let us describe these concepts more formally.

\begin{definition}\label{ALdef}
An operator satisfies Anderson localization (AL) if it has pure point spectrum with exponentially decaying eigenfunctions.
\end{definition}

Consider now the  Schr\"{o}dinger equation:
\begin{equation}\label{seq}
H (\x) \psi  =  E \psi 
\end{equation}
for $\psi = [ \psi_n ]_{n \in \mathbb{Z}} \subset  \mathbb{R}$ and $E \in \R$.

Due to Schnol-Simon's theorem (see \cite{CFKS}), to prove AL it is enough to show that every extended state is exponentially decaying. In other words, if 
$\psi$ is a formal solution to the  Schr\"{o}dinger equation (\ref{seq}) and if $\psi$ grows at most polynomially $\abs{\psi_n} \less \sabs{n}$ then $\psi$ decays exponentially: $\abs{\psi_n}  \less e^{-c \sabs{n}}$

\smallskip

The Schr\"{o}dinger equation (\ref{seq}) is a second order finite differences equation:
$$ - \psi_{n+1} - \psi_{n-1} + \la \, v (\shift^n \x) \, \psi_n = E \, \psi_n$$
which becomes 
$$
\Bigl[\begin{array}{cc}
\psi_{n+1}\\
\psi_n \\  \end{array} \Bigr]  = M_{N} (\x, E)  \Bigl[ \begin{array}{cc}
\psi_1\\
\psi_0 \\  \end{array} \Bigr]$$
where $$ M_{N} (\x, E) = M_{N} (\x, \la, E)  :=
\prod_{j=N}^{1}  \Bigl[ \begin{array}{ccc}
\lambda v (\shift^j \, \x ) - E  & &  - 1  \\
1 & &  0 \\  \end{array} \Bigr] $$ 
is called the transfer (or fundamental) matrix of  (\ref{op1}).

Define further the function
$$L_{N} (\x, E) = L_{N} (\x, \la, E)  := \frac{1}{N} \log \norm{ M_N (\x, E) }$$ 
and its mean
$$L_N (E) = L_{N} (\la, E) := \int_{\T^2} \frac{1}{N} \log \norm{ M_N (\x, E) } \, d\x$$ 
Due to sub-additivity, the sequence $L_N (E)$ converges.

\begin{definition}\label{lyapdef}The limit 
$$L(E) := \lim_{N \rightarrow \infty} L_N (E) $$
is called the Lyapunov exponent of  (\ref{op1})  and it measures the average exponential growth of the transfer matrices.
\end{definition}

Ergodicity in fact implies that for a.e. $\x \in \T^2$,
\begin{equation}\label{ldt-ergod}
L(E) := \lim_{N \rightarrow \infty} \frac{1}{N} \log \norm{ M_{N}(\x, E) }
\end{equation}

Note that since the transfer matrices have determinant \(1\), the Lyapunov exponent is always nonnegative. An important question is whether it is in fact (uniformly) bounded away from \(0\). This would imply, due to Kotani's theorem, absence of absolutely continuous spectrum, and it would represent a strong indication for pure point spectrum. This is also usually the assumption under which strong continuity properties hold and an analysis of the topological structure of the spectrum is more feasible.

\medskip

In this paper we prove Anderson localization and positivity and  continuity of the Lyapunov exponent for large coupling constant,  for certain ergodic transformations $\shift$ on $\T^2$ and under certain regularity and transversality conditions on the potential function $v (\x)$.
While all of our results are stated and proven for the two-dimensional torus $\T^2$, their analogues on the higher dimensional torus hold as well.

\medskip

We describe the assumptions on the transformation and on the potential function. 

\smallskip

We start with some notations: for a multi-index $\m = \mm \in \Z^2$, we write $\sabs{\m} := \sabs{m_1} + \sabs{m_2}$ and if $\m = \mm \in \N^2$, then $\m! := m_1! \cdot m_2!$
Moreover, for $\alfa = (\alpha_1, \alpha_2) \text{ and }  \m = \mm \in \N^2$, we write $\alfa \le \m$ when $\alpha_1 \le m_1$ and $\alpha_2 \le m_2$.

\medskip

Throughout this paper, the transformation $\shift \colon \T^2 \to \T^2$ will represent:

$\blob$ Either the \emph{skew-shift} 
 \begin{equation}\label{skew}
{\rm S}_\om \, \xx := (x_1 + x_2,  x_2  + \om)
\end{equation}
where $\om \in \T$ is irrational.

Its $n$th iteration is given by:
\begin{equation}\label{skewn}
{\rm S}_\om^n \, \xx = (x_1 + n x_2 + \frac{n (n-1)}{2} \om, \, x_2 + n \om)
\end{equation}

$\blob$ Or the \emph{multi-frequency shift} 
 \begin{equation}\label{multishift}
\shift_{\omm} \, \xx := (x_1 + \om_1,  x_2 + \om_2)
\end{equation}
where $\omm = (\om_1, \om_2) \in \T^2$ and $\om_1,  \om_2$ are rationally independent.

Its $n$th iteration is given by:
\begin{equation}\label{multishiftn}
\shift_{\omm}^n \, (\x) = \x + n \, \omm = (x_1 + n \om_1, \, x_2 + n \om_2)
\end{equation}

The irrationality /  rational independence of the frequency ensures that the corresponding transformation is ergodic. However, we need to make a quantitative assumption on this rational independence, for reasons that will be described later.

\smallskip

We say that the frequency $\om \in \T $ satisfies a Diophantine condition $DC_\ka$  for some $ \kappa > 0 $ if 
\begin{equation}\label{DC}
\mbox{dist } (l \omega,  \mathbb{Z} )  =:  \norm{l \, \om} > \kappa \cdot 
\frac{1} {\sabs{l} [ \log(1 + \sabs{l}) ]^{2}}   \quad \text{ for all } \  l \in  \mathbb{Z} \setminus \{ 0 \} 
\end{equation}

We say that the (multi)frequency $\omm \in \T^2 $ satisfies a Diophantine condition $DC_\ka$  for some $ \kappa > 0 $ and a fixed constant $A >2$, if 
\begin{equation}\label{DCM}
 \norm{\li \cdot \omm} := \norm{l_1 \,  \om_1 + l_2 \, \om_2} > \kappa \cdot 
\frac{1} {\sabs{\li}^A}  \quad \text{ for all } \  \li \in  \mathbb{Z}^2 \setminus \{ (0, 0) \} 
\end{equation}

Note that the set of frequencies which satisfy either \eqref{DC} or \eqref{DCM} has measure $1 - \text{O} (\ka)$, hence almost every frequency satisfies such a Diophantine condition $DC_\ka$ for some $\ka > 0$.

When necessary, to emphasize the dependence of the operator on the frequency, we will use the notation $H_{\om} (\x)$ or $H_{\omm} (\x)$ respectively.

\medskip

Now we describe the assumptions on the potential function $v (\x)$.

$\blob$ We say that a $C^\infty$ function $v (\x)$ on $\T^2$ belongs to the Gevrey class $G^{s} (\T^2) $ for some  $ s > 1 $ if its partial derivatives have the following bounds:
\begin{equation}\label{GC}
\sup_{ \x \in \mathbb{T}^2}  \abs{ \partial ^{\m} \, \, v(\x) } \leq M  K^{|\m|} ( \m ! )^{s}  \quad
\text{ for all } \quad  \m \in \N^2
\end{equation}
for some constants $M,$ $K$ $ > 0 $.

This condition is equivalent (see the exercises from Chapter IV in \cite{Ka}) to the following exponential-type decay of the Fourier coefficients of $v$:
\begin{equation}\label{fcoef} 
\abs{ \hat{v} (\li) } \leq  M e^{- \rho \sabs{ \li }^{1/s}} \quad  \text{ for all } \quad  \li \in \mathbb{Z}^2
\end{equation} 
for some constants $M,$ $\rho$ $ > 0, $ where 
$\displaystyle 
v(\x) = \sum_{\li \in \mathbb{Z}^2} \hat{v} (\li) \, e^{2 \pi i \, \li \cdot \x}
$ 

Note from (\ref{GC}) or (\ref{fcoef}) with $s = 1$ that the Gevrey class $ G^{1} (\mathbb{T}^2)$ is the class of real analytic functions on $\mathbb{T}^2 $.

Note also that $ s_1  <  s_2 \hspace{.1in} \Rightarrow \hspace{.1in} G^{s_1} (\mathbb{T}^2) \subset G^{s_2} (\mathbb{T}^2) $, so the greater the order of the Gevrey class, the larger the class. 

The Gevrey-class of any order $s > 1$ is an intermediate Carleman class of functions between analytic functions and  $C^\infty$ functions. They are not, however, quasi-analytic (one can easily construct examples or use a general test for quasi-analyticity of Carleman classes, as in Chapter V.2 in \cite{Ka}). 

We will then impose on our potential function $v$ the following generic transversality condition (TC).

\smallskip

$\blob$ We say that a function $v (\x)$ is transversal if $v$ is not flat at any point:
 \begin{equation}\label{TC}
\text{For any } \,  \x \in \T^2 \  \text{ there is }  \,  \m \in \N^2, \,  \sabs{\m}  \neq 0  \ \text{ such that }  \  \partial ^{\m} \, \, v(\x) \neq 0
\end{equation}

\smallskip

Non-constant analytic functions automatically satisfy \eqref{TC}. Therefore, a Schr\"{o}dinger operator with  potential given by a function which satisfies the Gevrey regularity condition \eqref{GC} and the transversality condition \eqref{TC} is a natural extension of the non constant analytic case considered in \cite{BG}, \cite{BGS}.

\medskip

We are ready to formulate the main result of this paper. 

\begin{theorem}\label{main}
Consider the Schr\"{o}dinger operator (\ref{op1}) on $l^2(\Z)$:
$$
[H (\x) \, \psi]_n := - \psi_{n+1} - \psi_{n-1} + \la \, v (\shift^n \x) \, \psi_n
$$
where the transformation $\shift$ is either the skew-shift \eqref{skew} or the multi-frequency shift \eqref{multishift}.
Assume that for some $\ka > 0$ the underlying frequency satisfies the Diophantine condition $DC_\ka$ described in \eqref{DC} or \eqref{DCM} respectively.

Assume moreover that the potential function $v (\x)$ belongs to a Gevrey class $G^s (\T^2)$ and that it is transversal as in \eqref{TC}. 

\smallskip

There is $\la_{0} = \la_{0} (v, \ka)$ such that the following hold:

\smallskip
 
$\blob$ If $ \sabs{\la}  \geq \la_{0},$ the Lyapunov exponent is positive for all energies $E \in \mathbb{R} $:
\begin{equation}\label{lyap1} 
L (E) \geq \frac{1}{4} \log \sabs{\la}  > 0 
\end{equation}

$\blob$ If $  \sabs{\la} \geq \la_{0},$ the Lyapunov exponent $L (E)$ is a  continuous functions of the energy $E$, with modulus of continuity on any compact interval $\mathcal{E}$  at least:
\begin{equation}\label{modcont}
h (t) = C \, e^{- c  \sabs{\log t}^{\eta}}
\end{equation}
where $ C = C( \mathcal{E}, \la, v, \ka, s) $ and $c$, $\eta$ are some positive absolute constants.

\smallskip

$\blob$  Let $\shift = {\rm S}_\om $ be the skew-shift. For every $\la$ with $\sabs{\la} \geq \la_{0}$, there is an exceptional set $\Bad = \Bad_\la \subset \T^3$, with $mes \,  \Bad < \kappa$, such that for all $(\om, \x) \notin \Bad$,  the operator $H_{\om} (\x)$ satisfies Anderson localization. 

\smallskip
 
$\blob$  Let $\shift = \shift_{\omm} $ be the multi-frequency shift. Fix $ \x_0 \in \mathbb{T}^2 $ and $\la$ with $ \sabs{\la} \geq \la_{0}$. Then for a.e. multi-frequency  $\omm \in DC_\ka$, the operator $H_{\omm} (\x_0) $ satisfies Anderson localization.

\end{theorem}

\bigskip

\section{Summary of related results, general strategy}\label{strategy} 
The results in this paper extend the ones in \cite{BGS} and \cite{BG} (see also J. Bourgain's monograph \cite{B}) from non-constant real analytic potential functions, to the more general class of Gevrey potential functions satisfying  a transversality condition.
They also mirror similar results obtained for the one-frequency shift model on the torus $\T$ (see \cite{sK1}). 

It should be noted, however, that unlike the one or multi-frequency shift,  the skew-shift, due to its weekly mixing properties, is expected to behave more like the random model (presumably regardless of the regularity of the potential). In other words, for the skew-shift, these results are expected to be independent of the size of the disorder $\la$. Hence one expects that if $\la \neq 0$, the Lyapunov exponent is positive and Anderson localization holds for all energies. Moreover, one expects no gaps in the spectrum (unlike in the one-frequency shift case, when the spectrum is a Cantor set) - see the comments at the end of Chapter 15 in \cite{B}. Some results on these very challenging problems have been obtained in  
\cite{B1}, \cite{B2}, \cite{hK1}, \cite{hK2}.

Localization results for these types of operators defined by skew-shift dynamics have applications to quantum chaos problems. More specifically, they imply existence of almost periodic solutions to the quantum kicked rotator equation. However, one has to establish (dynamical) localization for a more general, long range operator, one where the discrete Laplacian is replaced by a Toeplitz operator with fast off-diagonal decay of its monodromy matrix entries. This was already established for analytic potential functions (see Chapter 15 and 16 in \cite{B}), but we will not address this problem for Gevrey potential functions in this paper. 

Most of the results on localization for discrete quasiperiodic Schr\"{o}dinger operators (with either shift or skew-shift dynamics) have been obtained under the assumption that the potential function is the cosine function, or a trigonometric polynomial or a real analytic and non-constant function (see J. Bourgain's monograph \cite{B}).
 
Assuming Gevrey regularity and a transversality condition, there are localization results for the shift model that closely resemble the ones in the analytic case (see \cite{E}, \cite{sK1}). It should be noted, however, that they are usually perturbative and that more subtle results regarding fine continuity properties of the integrated density of states (as in \cite{GS1}, \cite{GS2}) or the topological structure of the spectrum (as in \cite{GS3}) are not available in this context.   

For potential functions that are more general than Gevrey, namely $C^\alpha$, the results available now (on localization and positivity of the Lyapunov exponent) require that  a(n asymptotically small relative to the size $\la$ of the disorder but)  positive set of energies be excluded or that the potential function be replaced by some generic variations of itself (see  \cite{kB}, \cite{jC}, \cite{CGS}).

\medskip

To prove Theorem \ref{main}  we will follow the same strategy used in \cite{sK1} for the single frequency shift model: at each scale, substitute the potential function by an appropriate  polynomial approximation (see Section~\ref{approximation}). This in turn will allow the use of subharmonic functions techniques (see Section~\ref{averages}) developed in \cite{BG}, \cite{BGS}, \cite{B}. An additional challenge is describing the transversality condition \eqref{TC} for multi-variable smooth functions in a quantitive way. We derive (see Section~\ref{lojasiewicz}) a {\L}ojasiewicz type inequality for such functions, of the kind previously available for non-constant trigonometric polynomials  (see \cite{jK}, \cite{DK}) or analytic functions (see \cite{PSS}, \cite{GS1}).      

\medskip

 The main technical result of this paper, from which all statements in Theorem~\ref{main} follow,  is a large deviation theorem (LDT) for logarithmic averages of  transfer matrices (see Section~\ref{ldt_proofs}).
 
 According to (\ref{ldt-ergod}), due to ergodicity, for a.e. $\x \in \T^2$:
$$ \frac{1}{N} \log \norm{ M_{N}(\x, E) } \to L(E) \ \text{ as } N \to \infty$$

The LDT provides a quantitative version of this convergence:
\begin{equation}\label{ldt-idea}
\mbox{mes } [ \x \in \T^2: \abs{ \frac{1}{N} \log \norm{ M_{N}(\x, E) }- L_{N} (E) } > \epsilon ] <  \delta (N, \epsilon)
\end{equation}
where $\epsilon = o (1) $ and $\delta (N, \epsilon) \rightarrow 0$ as $N \rightarrow \infty$

\medskip

The size of the deviation $\epsilon$ and the measure of the exceptional set $\delta (N, \epsilon)$ are very important. The sharpest such estimate (see Theorem 7.1  in \cite{GS1}), available for the single-frequency shift model with analytic potential, holds for any $\epsilon >0$ and exponentially small measure  $\delta (N, \epsilon) \approx e^{-c \delta N}$, thus morally matching the large deviation estimates for random variables that these deterministic quantities mimic here. Having such sharp estimates leads to a sharper modulus of continuity of the Lyapunov exponent (see \cite{GS1}).

For the multi-frequency shift and the skew-shift models, even with analytic potentials, the available estimates are not as sharp. 
In this paper, for Gevrey potential functions, we will get $\epsilon \approx N^{-\tau}$ and $\delta \approx e^{-N^\sigma}$ for some absolute constants $\tau, \sigma \in (0,1)$.

\medskip

Following the approach in \cite{B}, \cite{BGS}, a large deviation estimate like (\ref{ldt-idea}) will allow us to obtain a lower (positive) bound and continuity of the Lyapunov exponent, once these properties are established at an initial scale $N_0$ for $L_{N_0} (E)$. It will also allow us (the reader will be refered to \cite{B}, \cite{BGS} for details) to establish estimates on the Green's functions associated with the operator (\ref{op1}), more specifically the fact that double resonances for Green's functions occur with small probability, which leads to Anderson localization.

Most of the paper will then be devoted to proving a  LDT like (\ref{ldt-idea}):
 \begin{equation}\label{ldt-strategy}
\mbox{mes } [ \x \in \T^2: | \frac{1}{N} \log \norm{ M_{N}(\x, E) }- L_{N} (E) | > N^{-\tau} ] <  e^{-N^\sigma}
\end{equation}  
through an inductive process on the scale $N$.

\smallskip

The \textit{base step} of the inductive process for proving the LDT (\ref{ldt-strategy}) is based exclusively on the transversality condition \eqref{TC} on  the potential, and on choosing a sufficiently large disorder $\la$. The latter is what makes this approach perturbative (and, in the case of the skew-shift model,  wasteful, since it does not exploit the weakly-mixing properties of its dynamics). The former implies a {\L}ojasiewicz type inequality, which we prove using a quantitative form of the implicit function theorem.

\smallskip

In the \textit{inductive step} we use the regularity of the potential function $v (\x)$ and the arithmetic properties of the frequency. The regularity of $v (\x)$ allows us to approximate it efficiently by trigonometric polynomials $v_N (\x) $ at each scale $N,$ and to use these approximants in place of $v (\x)$ to get analytic substitutes $\tilde{M_N} (\x)$ for the transfer matrices $M_N (\x)$. Their corresponding logarithmic averages will be subharmonic in each variable which will allow us to employ the subharmonic functions techniques developed in \cite{B}, \cite{BG}, \cite{BGS}.

The main technical difficulty with this approach, and what restricts it to Gevrey (instead of say, $C^\alpha$) potential functions, is that the holomorphic extensions of the transfer matrix substitutes  $\tilde{M_N} (\x)$ will have to be restricted to domains of size $\approx N^{-\delta}$ for some $\delta > 0$. In other words, the estimates will not be uniform in $N$, and this decreasing width of the domain of holomorphicity will have to be overpowered. This will not be possible for a $C^\alpha$ potential function because its  trigonometric polynomial approximation is less efficient, so the width of holomorphicity in this case will decrease too fast (exponentially fast).

\smallskip

\section{ Description of the approximation process}\label{approximation} 
Let $v \in G^s (\T^2)$ be a Gevrey potential function. Then
\begin{equation}\label{fourierexp}
v(\x) = \sum_{\li \in \mathbb{Z}^2} \hat{v} (\li) e^{2 \pi i \, \li \cdot \x}
\end{equation} 
where for some constants $M, \rho > 0,$ its  Fourier coefficients have the decay:
\begin{equation}\label{fouriercoef} 
\abs{  \hat{v} (\li) } \leq  M e^{- \rho | \li |^{1/s}}  \ \text{ for all } \  \li \in \mathbb{Z}^2
\end{equation}

We will compare the logarithmic averages of the transfer matrix  
\begin{equation}\label{LNx}
L_N (\x, E) = \frac{1}{N} \log \norm{ M_N (\x, E) } \, d\x = \frac{1}{N} \, \log \norm{ \prod_{j=N}^{1}  \Bigl[ \begin{array}{ccc}
\lambda v (\shift^j \, \x ) - E  &   - 1  \\
1 &   0 \\  \end{array} \Bigr] }
\end{equation}
with their means
\begin{equation}\label{LN}
L_{N} (E) = \int_{\T^2} L_N (\x, E)  \, d\x
\end{equation}

To be able to use subharmonic functions techniques, we will have to approximate the potential function $v (\x)$ by trigonometric polynomials $v_N (\x)$ and substitute $v$ by $v_N$ into  (\ref{LNx}). At each scale $N$ we will have a different approximant   $v_N$ chosen in such a way  that the ``transfer matrix substitute'' would be close to the original transfer matrix. The approximant $v_N$ will then have to differ from $v$ by a very small error - (super)exponentially small in $N$. That, in turn, will make the degree deg $v_N =: \tilde{N}$ of this polynomial very large - based on the rate of decay (\ref{fcoef}) of the Fourier coefficients of $v$, $\tilde{N}$ should be a power of $N$, dependent on the Gevrey class $s$.

The trigonometric polynomial $v_N (\x)$ has an extension $v_N (\z)$, $\z = \zz$, which is separately holomorphic on the whole complex plane in each variable. We have to restrict $v_N (\z)$  in each variable to a narrow strip (or annulus, if we identify the torus $\T$ with $\R/\Z$)  of width $\rho_N$, where  $\rho_N \approx ( \mbox{deg }v_N )^{- 1} \approx \tilde{N}^{-1} \approx N^{- \theta}$, for some power $\theta > 0$. 
This is needed in order to get a uniform in $N$ bound on the extension $v_N (\z)$. Moreover, in the case of the skew-shift, this is also needed because its dynamics expands in the imaginary direction, and in this case, the width of holomorphicity in the second variable will have to be smaller than in the first by a factor of $\approx \frac{1}{N}$. 

The fact that the ``substitutes''  $v_N (\x) $ have different, smaller and smaller widths of holomorphicity creates significant technical problems compared to the case when $v(\x)$ is a real analytic function. It also makes this approach fail when the rate of decay of the Fourier coefficients of the potential function $v (\x)$ is slower.

Therefore, we have to find the optimal ``error vs. degree'' approximations of $v(\x)$ by trigonometric polynomials $v_N (\x) $. Here are the formal calculations.
  
\smallskip

For every positive integer $N$, consider the truncation
\begin{equation}\label{trunc}
v_N (\x) := \sum_{| \li | \leq \tilde{N}} \hat{v} (\li) \, e^{2 \pi i \, \li \cdot  \x}
\end{equation} 
where $\tilde{N} = \mbox{deg } v_N$ will be determined later.

Since $v_N \xx $ is in each variable a $1$-periodic, real analytic function on $\mathbb{R}$, it can be extended to a separately in each variable $1$-periodic holomorphic function on $\mathbb{C}$: 
\begin{equation}\label{holext} 
v_N (\z) := \sum_{\sabs{ \li } \leq \tilde{N}} \hat{v} (\li) e^{2 \pi i \, \li \cdot \z}
\end{equation}

To ensure the uniform boundedness in $N$ of $v_N \zz$ we have to restrict 
 $v_N \zz$ to the  annulus/strip  $[ \sabs{ \Im z  } < \rho_{1,N} ] \times [ \sabs{ \Im z  } < \rho_{1,N} ] $, where $$ \rho_{1,N} := \frac{\rho}{2} \tilde{N}^{ - 1 + 1/s} $$

 Indeed, if $z_1 = x_1 + iy_1$,  $z_1 = x_1 + iy_1$ and $\sabs{y_1}, \sabs{y_2} < \rho_{1,N}$, then: 

\begin{align*}
\abs{ v_N \zz } = \abs{ \sum_{| \li | \leq \tilde{N}} \hat{v} (\li) e^{2 \pi i \, \li \cdot  \z} }  \le & \, 
\sum_{| \li | \leq \tilde{N}} \, \abs{\hat{v} (\li) } e^{- 2 \pi \, \li  \cdot \y } \\
 \leq M \sum_{| \li | \leq \tilde{N}} e^{- \rho | \li |^{1/s}}  e^{| l_1 | | y_1 | + |l_2| |y_2|}
\le & \,
 M \sum_{| \li | \leq \tilde{N}} e^{- \rho | \li |^{1/s}}  e^{| \li | \, \rho_{1,N}}   \\
\le M \sum_{| \li | \leq \tilde{N}}  e^{- \rho | \li |^{1/s}} \cdot e^{| \li | \,  \rho/2 \, | \li |^{- 1 + 1/s}}  = & \, 
M \sum_{| \li | \leq \tilde{N}}  e^{- \frac{\rho}{2} | \li |^{1/s}}   \\
\le M \sum_{\li \in \Z^2}  e^{- \frac{\rho}{2} | \li |^{1/s}}  =: & \, B < \infty 
\end{align*}
where $B$ is a constant which depends on $v$ (not on the scale $N$) and we have used : 
$ \sabs{y_1}, \sabs{y_2}  <  \rho_{1, N} =  \frac{\rho}{2} \tilde{N}^{- 1 + 1/s} \leq  \frac{\rho}{2} | \li |^{- 1 + 1/s} $  for $ | \li | \leq \tilde{N}$, since $s > 1$.

\smallskip

We also clearly have $ | v (\x) - v_N (\x)| \lesssim e^{-\rho \tilde{N}^{1/s}}  $  for all 
$\x \in \T^2$.

\smallskip

We will need, as mentioned above, super-exponentially small error in how $v_N (\x)$ approximates $v (\x)$, otherwise the error would propagate and the transfer matrix substitutes will not be close to the original transfer matrices. Hence $\tilde{N}$ should be chosen such that say $ e^{-\rho \tilde{N}^{1/s}} \leq e^{- \rho N^{2}}.$
So if $\tilde{N} := N^{2 s}$, then the width of the holomorphic  (in each variable) extension $v_N (\z)$  will be $\rho_{1,N}  = \frac{\rho}{2} N^{2  s  (- 1+ \frac{1}{s})} =  \frac{\rho}{2}  N^{- 2 (s - 1)}
=: \frac{\rho}{2} N^{- \delta}$, where $\delta := 2 \, (s - 1) > 0$. 

\smallskip

We conclude: for every integer $N \geq 1$, we have a function $v_N (\x)$ on $\mathbb{T}^2$ such that
\begin{equation}\label{aproxtrunc} 
\abs{ v (\x) - v_N (\x) } < e^{- \rho N^2} 
\end{equation}
and $v_N (\x)$ has a $1$-periodic separately holomorphic extension $v_N (\z)$ to the strip $[ \sabs{ \Im z  } < \rho_{1,N} ] \times [ \sabs{ \Im z  } < \rho_{1,N} ] $, where $ \rho_{1,N}  = \frac{\rho}{2} N^{- \delta} $,   for which 
\begin{equation}\label{boundv} 
\abs{ v_N (\z) } \leq B
\end{equation} 
The positive constants $\rho$, $B$, $\delta$ above depend only on $v$ (not on the scale $N$).
The constant $\delta$ depends on the Gevrey class of $v$:  $\delta := 2 (s - 1) $ so it is fixed but presumably very large. 
\medskip

We now substitute these approximants $v_N (\x)$ for $v(\x)$ in the definition of the transfer matrix $M_N (\x)$. 

Let  $$ A(\x, E)  
 :=  \Bigl[ \begin{array}{ccc}
\lambda v (\x) - E  & &   - 1  \\
1 & &   0 \\  \end{array} \Bigr]$$
be the cocycle  that defines the transfer matrix  $M_N (\x)$.  

Consider then
$$ \tilde{A}_N (\x, E) 
 :=  \Bigl[ \begin{array}{ccc}
\lambda v_N (\x) - E  & &   - 1  \\
1 & &   0 \\  \end{array} \Bigr]$$
which leads to the transfer matrix substitutes
$$\tilde{M} _{N} (\x, E)  := \prod_{j=N}^{1} \tilde{A}_N (\shift^{j} \x, E) $$

To show that the substitutes are close to the original matrices, we use Trotter's formula. This is a wasteful approach, and clearly in part responsible for our inability to apply these methods beyond Gevrey functions. There are other, much more subtle reasons for why this approach is limited to this class of functions. 

$$ M_{N} (\x) -  \tilde{M}_{N} (\x)  =  $$ 
$$ = \sum_{j=1}^{N} A (\shift^N \x) \ldots A (\shift^{j+1} \x) \,  [A (\shift^j \x) - \tilde{A}_{N}(\shift^j \x)] \,  \tilde{A}_{N}(T^{j-1} \x) \ldots \tilde{A}_{N} (\shift \x)$$
$$ A (\shift^j \x) - \tilde{A}_{N}(\shift^j \x) = \left[ \begin{array}{cc}
\lambda v (\shift^j \x ) - \lambda v_{N} (\shift^j \x )  &   0  \\
0 &   0 \\  \end{array} \right ]$$ so 
$$ \norm{ A (\shift^j \x) - \tilde{A}_{N}(\shift^j \x) } \le \sabs{\la} \, \sup_{\y \in \mathbb{T}^2} \abs{ v (\y) - v_N (\y) }  <   \sabs{\la} \, e^{-\rho N^{2}} $$ 

\smallskip

Since $  \sup_{\x \in \mathbb{T}^2} | v (\x) | \leq B $, the spectrum of the operator $ H (\x) $ is contained in the interval $ [ -2 - \sabs{\la} \, B,  2 + \sabs{\la} \, B \, ]$. Hence it is enough to consider only the energies $E$ such that $ |E| \leq 2 + |\lambda| \, B $. We then have:

$$ \norm{ A (\shift^j \x) }  = \bnorm{  \Bigl[ \begin{array}{cc}
\lambda v (\shift^j \x ) - E  & - 1  \\
1 & 0 \\  \end{array} \Bigr] }  \le  \sabs{\la} \, B  + \abs{E} + 2  \le  2 \sabs{\la}  B + 4  \le  e^{S(\la)} $$

and

$$ \norm{ \tilde{A}_{N} (\shift^j \x) }  \le  \bnorm{  \Bigl[ \begin{array}{cc}
\la v_{N} (\shift^j \x ) - E  & - 1 \\
1 & 0 \\  \end{array} \Bigr] }  \le \sabs{\la}  \, B + \abs{E} + 2  \leq e^{S(\lambda)} $$ 

Therefore,
\begin{equation}\label{boundA}
 \norm{ A (\shift^j \x) }, \, \norm{ \tilde{A}_{N} (\shift^j \x) }   \leq  e^{S(\la)}
\end{equation}
where $S(\lambda) \approx \log \sabs{\la}$ is a scaling factor that depends only on the (assumed large)  disorder $\la$ and on  $v$ (the constants inherent in  $\approx$ depend on the number $B = B(v)$ which also determines the range of spectral values  $E$). 

\smallskip

We then have:
$$ \norm{ M_{N} (\x, E) - \tilde{M}_{N} (\x, E) }  \le  \sum_{j=1}^{N} e^{S(\lambda)} 
\ldots  e^{S(\lambda)} \,  | \la | \, e^{-\rho N^2}  e^{S(\lambda)} \ldots  e^{S(\lambda)} \leq  $$
$$\leq e^{N S(\lambda) - \rho N^2}  \leq  e^{-\rover N^2}$$
provided $ N \gtrsim S(\lambda)$.

Hence uniformly in $\x \in \T^2$ we get:
\begin{equation}\label{tmclose}
 \norm{  M_{N} (\x, E) - \tilde{M}_{N} (\x, E) } \le e^{-\rover N^2}
 \end{equation}
provided we choose 
\begin{equation}\label{n>b(la)}
N \gtrsim S(\lambda)
\end{equation}
which means roughly that $\la$ has to be at most exponential in the scale $N$.

\medskip

We are now going to turn our attention to the logarithmic averages of the transfer matrices.

Since $\det M_N (\x) = 1 $ and  $\det \tilde{M}_N (\x) = 1 $, we have that 
$ \norm{ M_N (\x) } \ge  1 $ and $ \norm{ \tilde{M}_N (\x) } \ge  1 $. Thus, for all $ N \gtrsim S(\lambda)$ and for every $\x \in \T^2$, 

$$ \abs{ \frac{1}{N} \log \norm{ M_{N} (\x) }  -  \frac{1}{N} \log \norm{ \tilde{M}_{N} (\x) } }  \le  \frac{1}{N} \norm{ M_{N} (\x)  - \tilde{M} _{N} (\x) }  < e^{-\rover N^2} $$

Recall the following notation:
\begin{equation}\label{LNxE}
L_N (\x, E) = \frac{1}{N} \log \norm{ M_N (\x, E) } \, d\x 
\end{equation}
and define its substitute:
\begin{equation}\label{shN}
u_{N} (\x, E)  := \frac{1}{N} \log || \tilde{M}_{N} (\x) || 
\end{equation}

Therefore, uniformly in $\x \in \T^2$ and in the energy $E$:
$$\abs{ L_N (\x, E)  - u_{N} (\x, E) }  < e^{-\rover N^2} $$
and by averaging in $\x$: 
$$ \abs{ L_N (E) - \avg{u_N (E)} }  < e^{-\rover N^2}$$
where $ L_{N} (E) := \int_{\T^2} L_N (\x, E)  \, d\x$ and for any function $u (\x)$, $\ \avg{u} := \int_{\T^2} u (\x) \, d \x$.

\medskip

The advantage of the substitutes $u_{N} (\x)$ is that they extend to pluri-subharmonic functions in a neighborhood of the torus $\T^2$, as explained below.

\smallskip

For the skew-shift transformation $\shift = \skews$ we consider the strip  $ \Aa_{\ro_N} := [ \sabs{ \Im z  } < \rho_{1,N} ] \times [ \sabs{ \Im z  } < \rho_{2,N} ] $ where $ \rho_{1,N} = \roverr N^{- \delta}$ and 
 $ \rho_{2,N} := \frac{\rho_{1,N}}{2 N}  = \roverr N^{- \delta -1}$ 
 
 \smallskip
 
We have to reduce the size of the strip in the second variable to account for the fact that the skew-shift expands in the imaginary direction. Our approximation method required a reduction in the size of the holomorphicity strip at each scale, and this additional reduction will be comparatively harmless.

If we extend the map $\skews$ from $\T^2 = (\R/\Z)^2$ to $\C^2$, by 
$$\displaystyle \skews \zz = (z_1+z_2, z_2 + \om)$$ we get as in (\ref{skewn}) that 
$$\displaystyle 
\skews^n \zz = (z_1 + n z_2 + \frac{n (n-1)}{2} \om, z_2 + n \om)
$$

Then if $\zz \in  \Aa_{\ro_N} $ and if we perform $n \leq N$ iterations,  we have:
\begin{equation}\label{goodstrip}
\abs{  \Im (z_1 + n z_2 + \frac{n (n-1)}{2} \om) } = \abs{ \Im (z_1 + n z_2) } = \sabs{ y_1 + n y_2 } < \rover N^{-\delta}
\end{equation}

\smallskip

The matrix function $$\displaystyle  \tilde{A}_N (\x) =  \Bigl[ \begin{array}{ccc}
\lambda v_N (\x) - E  & &   - 1  \\
1 & &   0 \\  \end{array} \Bigr]$$  extends to a $1$-periodic, separately in each variable holomorphic matrix valued function: 
$$ \tilde{A}_N (\z) 
 :=  \Bigl[ \begin{array}{ccc}
\lambda v_N (\z) - E  & &   - 1  \\
1 & &   0 \\  \end{array} \Bigr]$$

Using \eqref{boundv} and the definition of the scaling factor $S (\la)$, we have that  on the strip $ \Aa_{\ro_N}$ the matrix valued function $ \tilde{A}_N (\z) $ is uniformly in $N$ bounded by $e^{S(\la)}$. 
Combining this  with (\ref{goodstrip}), the transfer matrix substitutes extend on the same strip to separately  holomorphic matrix valued functions 
$$\tilde{M} _{N} (\z, E)  := \prod_{j=N}^{1} \tilde{A}_N (\skews^{j} \, \z, E) $$
such that, for all $\z \in \Aa_{\ro_N} $ and for all energies $E$ we have
$$\norm{ \tilde{M} _{N} (\z, E)  } \le e^{N S(\la)}$$

Therefore,  
$$ u_{N} (\z)  := \frac{1}{N} \log \norm{ \tilde{M}_{N} (\z) } $$ 
is a pluri-subharmonic function on the strip $  \Aa_{\ro_N} $, and 
for any $\z$ in this strip, $$  \sabs{ u_N (\z) }  \le  S(\la) $$

\smallskip

The same argument applies to the multifrequency shift $\shift = \mshift$. The extension of this dynamics to the complex plane
$$\mshift \zz = (z_1 +  \om_1, z_2 +  \om_2)$$
does not expand in the imaginary direction, so there is no need to decrease the width of the strip in the second variable as in the case of the skew-shift. However, for convenience of notations, we will choose the same strip $\Aa_{\ro_N}$ for both transformations.

We can now summarize all of the above into the following.

\begin{lemma}\label{lemma1} 
For fixed parameters $\lambda, E$, for a fixed transformation $\shift = \skews$ or $\shift = \mshift$ and for $\delta = 2 (s-1)$, at every scale $N$ we have a 
$1$-periodic function 
$$u_{N} (\x)  := \frac{1}{N} \log \norm{ \tilde{M}_{N} (\x) } $$ 
which extends to a pluri-subharmonic function $ u_N (\z)$  on the strip  $  \Aa_{\ro_N} =  [ \left| \Im z \right| < \rho_{1, N} ]  \times [ \left| \Im z \right| < \rho_{2, N} ] $, where
$ \rho_{1,N}  \approx  N^{- \delta} $, $ \rho_{2,N}  \approx  N^{- \delta-1} $ 
so that 
\begin{equation}\label{boundu}
\sabs{ u_N (\z) }  \le  S(\la) \quad \text{ for all } \quad  
\z \in  \Aa_{\ro_N}
\end{equation}

Note that the bound (\ref{boundu}) is uniform in $N$. 

Moreover, if $ N \gtrsim  S(\lambda)$, then the logarithmic averages of the transfer matrices $ M_N (\x) $  are well approximated by their substitutes $ u_N (\x)$:
\begin{equation}\label{aproxu}  
\abs{ \frac{1}{N} \log \norm{ M_{N} (\x) }  -  u_N (\x) } \lesssim e^{- N^2}
\end{equation}
\begin{equation}\label{aprox<u>} 
\abs{ L_N -  \avg{u_N} }  \lesssim e^{- N^2}
\end{equation}
\end{lemma}

All the inherent constants in the above (and future) estimates are either universal or depend only on $v$ (and not on the scale $N$) so they can be ignored. The estimates above are independent of the variable $\x$, the parameters $\la, E$ and the transformation $\shift$.

This s a crucial technical result in our paper, which will allow us to use subharmonic functions techniques as in \cite{B},  \cite{BGS} for the functions $u_N$, and then transfer the relevant estimates to the rougher functions they substitute.

\medskip 
 
The logarithmic averages of the transfer matrix have an almost invariance (under the dynamics) property: 
 
\begin{lemma}\label{inv} 
For all $ \x \in \mathbb{T}^2,$ for all parameters $\la, E $ and for all transformations $\shift$ 
 we have : 
\begin{equation}\label{mshift} 
\abs{ \frac{1}{N} \log \norm{ M_{N} (\x) } -  \frac{1}{N} \log \norm{ M_{N} (\shift  \, \x) } \, } \lesssim  \frac{S(\la)}{N}
\end{equation}
\end{lemma}

\begin{proof}
\begin{align*}
 \abs{ \frac{1}{N} \log \norm{ M_{N} (\x) } -  \frac{1}{N} \log \norm{ M_{N} (\shift \, \x) } \, }  = \abs{ \frac{1}{N} \log \frac{\norm{M_{N} (\x) }}{\norm{ M_{N} (\shift \, \x) }}  \, } \\
 = \abs{ \frac{1}{N} \log 
\frac{\norm{ A (\shift^N \x ) \cdot \ldots \cdot A (\shift^2 \x ) \cdot  A (\shift \, \x ) } }
{\norm{ A ( \shift^{N + 1} \x ) \cdot  A (\shift^N \x ) \cdot \ldots \cdot A (\shift^2 \x ) } } \, }  \\
 \le \frac{1}{N} \log [ \, 
 \norm{ ( A(\shift^{N + 1} \x ) )^{- 1} } \cdot  \norm{ A (\shift \, \x ) } \, ] 
  \lesssim \frac{S (\la)}{N} 
\end{align*}
where the last bound is due to (\ref{boundA}). The inequality (\ref{mshift}) then follows.
\end{proof}

\bigskip

\section{Averages of shifts of pluri-subharmonic functions}\label{averages}
One of the main ingredients in the proof of the LDT \eqref{ldt-strategy} is an estimate on averages of shifts of pluri-subharmonic functions. These averages are shown to converge in a quantitative way to the mean of the function. The result holds for both the skew-shift and the multi-frequency shift. 

For the skew-shift, the result was proven in \cite{BGS} (see Lemma 2.6 there). We will reproduce here the scaled version of that result, the one that takes into account the size of the domain of subharmonicity and the sup norm of the function. The reader can verify, by following the details of the proof in \cite{BGS}, that this is indeed the correct scaled version.
For the multi-frequency shift, the result is essentially contained within the proof of Theorem 5.5 in \cite{B}, but for completeness, we will include here the details of its proof. 

\begin{proposition}\label{avshifts}
Let $ u (\x) $ be a real valued function on $\T^2$, that extends to a pluri-subharmonic function $u (\z)$ on a strip $\Aa_{\ro} = [ \sabs{ \Im z_1 } < \rho_1 ] \times   [ \sabs{ \Im z_2 } < \rho_2 ] $. Let $\rho = \min \{\rho_1, \rho_2 \}$.   Let $\shift$ be either the skew-shift or the multi-frequency shift on $\T^2$, where the underlying frequency satisfies the $DC_\kappa$ described in \eqref{DC} or \eqref{DCM} respectively.
Assume that
$$\sup_{\z \in \Aa_{\ro}} \, | u (\z) | \leq S$$ 
Then for some explicit  constants $\sigma_0, \tau_0 > 0$, and for $n \geq n (\kappa)$ we have:
\begin{equation}\label{shiftldt} 
\mbox{mes }[ \x \in \T^2 : \abs{ \frac{1}{n} \sum_{j = 0}^{n-1}  u (\shift^j \x)  \, -  \avg{u}   } > \frac{S}{\rho} \, n^{- \tau_0} ] < 
e^{- n^{\sigma_0}}
\end{equation}
\end{proposition}

Here is how this estimate can be understood. Given the ergodicity of the transformation $\shift$ for irrational (or rationally independent) frequencies, on the long run, the orbits  $\shift^j \x $ of most points $\x$ will tend to be fairly well distributed throughout the torus $\T^2$ (see the picture below). 

\begin{figure}[h]
\centering
\subfloat[Iterations of the skew-shift]{
\includegraphics[width=0.43\textwidth]{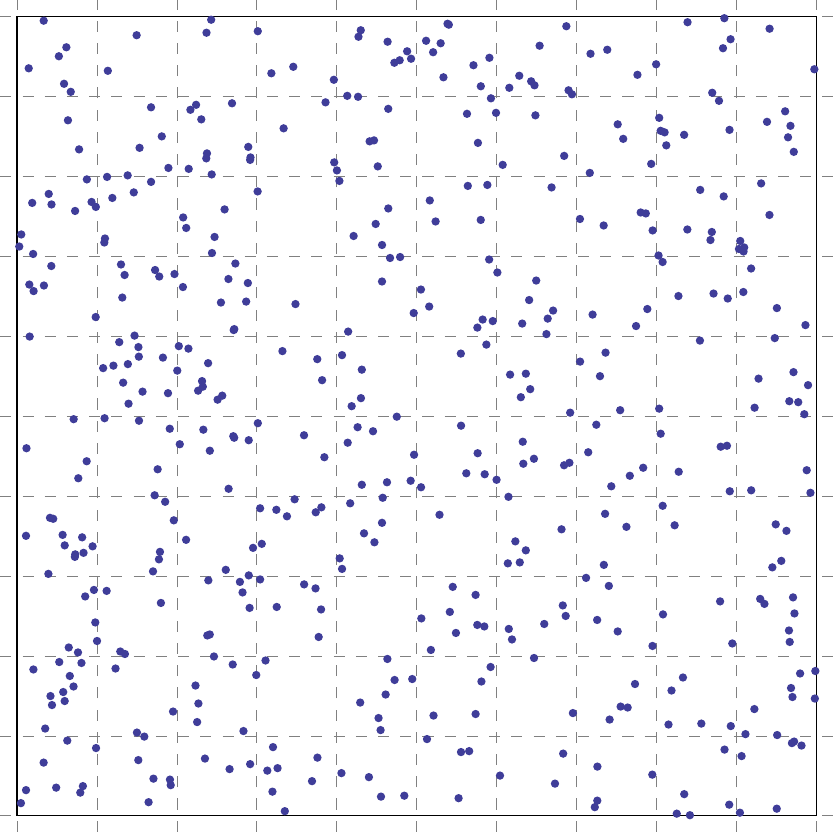}
\label{fig:skews}}
\qquad
\subfloat[Iterations of the multifrequency shift]{
\includegraphics[width=0.43\textwidth]{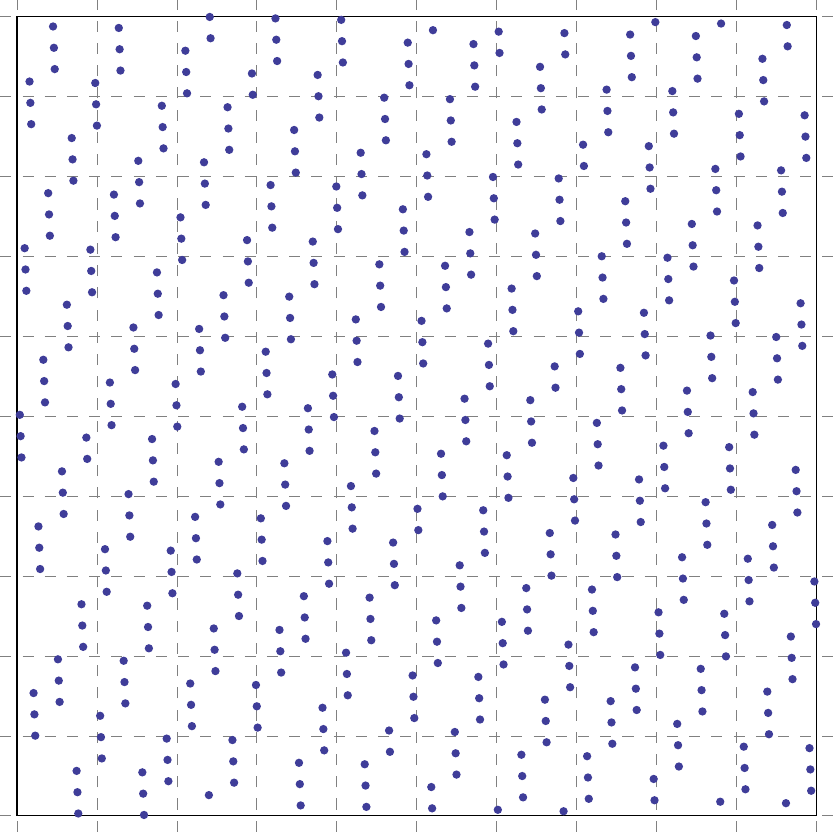}
\label{fig:mshift}}
\label{fig:globfig}
\end{figure}

The average $\ \frac{1}{n} \sum_{j = 0}^{n-1}  u (\shift^j \x) $ will then resemble a Riemann sum of the function $u (\x)$ and as such, it will approach the integral $\avg{u}$.

Moreover, a quantitative description of the irrationality (or rational independence) of the frequency in the form of a Diophantine condition like \eqref{DC}, \eqref{DCM},  should lead to a quantitative description of the convergence of the average sum to the integral $\avg{u}$.

To prove this quantitative convergence result,  we consider the Fourier expansion of  the function $u (\x)$ and apply it to the average sums. This leads to a convolution of $u (\x)$ with a Fej\'{e}r-type kernel. It is crucial to have estimates on the Fourier coefficients of the function $u$, and they are obtained via Riesz' representation theorem for subharmonic functions (see Corollary 4.1. in \cite{B}). Since $u \zz$ is pluri-subharmonic, the scaled version of Corollary 4.1. in \cite{B}  implies:
\begin{equation}\label{Riesz} 
\sup_{x_2 \in \T} \, \abs{\hat{u} (l_1, x_2)}  \less \frac{S}{\rho_1} \cdot \frac{1}{\sabs{l_1}}  \ \text{ and  } \ 
\sup_{x_1 \in \T} \, \abs{\hat{u} (x_1, l_2)}  \less \frac{S}{\rho_2} \cdot \frac{1}{\sabs{l_2}} 
\end{equation}

The estimates (\ref{Riesz}) imply (small) upper bounds on the $L^2$ - norm of the part of the Fourier expansion for which at least one of the indices $l_1$ and $l_2$ is large. The difficult part is when both indices $l_1$ and $l_2$  are small, in which case we use the Diophantine condition on the frequency to estimate the resulting exponential sums. 

In the case of the skew shift dynamics \eqref{skewn}, the resulting exponential sums are quadratic, and they are estimated using Weyl's method (see \cite{BGS} for the details of the proof).  We will now present the details of the proof for the multi-frequency shift case $\shift \, \x = \shift_{\omm} \, \x := \x + \omm$.

\begin{proof}
Expand $u (\x)$ into a Fourier series
$$u (\x) = \avg{u} + \sum_{\genfrac{}{}{0cm}{}{\li \in \Z^2}{\li \neq (0,0)}} \, \hat{u} (\li) \cdot e^{2 \pi i \, \li \cdot \x}$$
Then the averages of shifts have the form
\begin{align*}
\frac{1}{n} \sum_{j = 0}^{n-1}  u (\shift^j \x)  \ = \ & \frac{1}{n} \sum_{j = 0}^{n-1}  u (\x + j \omm) \\
= \ & \avg{u} +  \sum_{\genfrac{}{}{0cm}{}{\li \in \Z^2}{\li \neq (0,0)}} \, \hat{u} (\li) \cdot e^{2 \pi i \, \li \cdot \x} \cdot \Bigl( \frac{1}{n} \sum_{j = 0}^{n-1}  e^{2 \pi i \, j \, \li \cdot \omm} \Bigr )\\
= \ & \avg{u} +  \sum_{\genfrac{}{}{0cm}{}{\li \in \Z^2}{\li \neq (0,0)}} \, \hat{u} (\li) \cdot e^{2 \pi i \, \li \cdot \x} \cdot  K_n (\li \cdot \omm)
\end{align*}
where we denoted by $K_n (t)$ the Fej\'{e}r kernel 
$$K_n (t) = \frac{1}{n} \sum_{j = 0}^{n-1}  e^{2 \pi i \, j t} \, = \, \frac{1}{n} \, \frac{1 - e^{2 \pi i \, n t}}{1 - e^{2 \pi i \, t}}$$
which clearly has the bound
\begin{equation}\label{fejerkernelbound}
\abs{K_n (t)} \le \min \Bigl\{ 1, \frac{1}{n \norm{t}} \Bigr\}
\end{equation}
We then have:
\begin{align*}
\bnorm{ \frac{1}{n} \sum_{j = 0}^{n-1}  u (\x + j \omm) - \avg{u} }_{L^2(\T^2)}^2 \  = \ &
 \sum_{\genfrac{}{}{0cm}{}{\li \in \Z^2}{\li \neq (0,0)}} \, \abs{\hat{u} (\li)}^2 \cdot \abs{K_n (\li \cdot \omm)}^2\\
 =  \sum_{1 \le \sabs{\li} <  K} \, \abs{\hat{u} (\li)}^2 \cdot \abs{K_n (\li \cdot \omm)}^2 \  + \ &
  \sum_{\sabs{\li} \ge K} \, \abs{\hat{u} (\li)}^2 \cdot \abs{K_n (\li \cdot \omm)}^2
\end{align*}
We will estimate the second sum above using the bounds \eqref{Riesz} on the Fourier coefficients of $u (\x)$ and the first sum using the DC \eqref{DCM} on the frequency $\omm$. The splitting point $K$ will be chosen to optimize the sum of these estimates.

Clearly \eqref{Riesz} implies:
$$\sum_{l_2 \in \Z} \, \abs{\hat{u} \lli}^2 = \bnorm{\hat{u} (l_1, x_2)}_{L_{x_2}^2(\T)}^2 \less \, \Bigr( \frac{S}{\rho_1} \, \frac{1}{\sabs{l_1}}\Bigl)^2 \le  \, \Bigr( \frac{S}{\rho} \Bigl)^2 \, \frac{1}{\sabs{l_1}^2}$$
and
$$\sum_{l_1 \in \Z} \, \abs{\hat{u} \lli}^2 = \bnorm{\hat{u} (x_1, l_2)}_{L_{x_1}^2(\T)}^2 \less \, \Bigr( \frac{S}{\rho_2} \, \frac{1}{\sabs{l_2}}\Bigl)^2 \le  \, \Bigr( \frac{S}{\rho} \Bigl)^2 \, \frac{1}{\sabs{l_2}^2}$$

Then we have:
\begin{align*}
 \sum_{\sabs{\li} \ge K} \, \abs{\hat{u} (\li)}^2 \cdot \abs{K_n (\li \cdot \omm)}^2   \le \sum_{\sabs{\li} \ge K} \, \abs{\hat{u} (\li)}^2 \\
   \le  \sum_{\li \colon \sabs{l_1} \ge K/2} \, \abs{\hat{u} (\li)}^2  +    \sum_{\li \colon \sabs{l_2} \ge K/2}  \abs{\hat{u} (\li)}^2 
 \less  \,  \Bigr( \frac{S}{\rho} \Bigl)^2 \, \frac{1}{K}
\end{align*}

Estimate \eqref{Riesz} clearly impies:
$$\abs{\hat{u} (\li)} \less \frac{S}{\rho} \, \frac{1}{\sabs{\li}}$$
Then using the DC \eqref{DCM} on $\omm$ and \eqref{fejerkernelbound}, we obtain:
\begin{align*}
\sum_{1 \le \sabs{\li} <  K} \, \abs{\hat{u} (\li)}^2 \cdot \abs{K_n (\li \cdot \omm)}^2 \le 
\, \Bigr( \frac{S}{\rho} \Bigl)^2 \, \sum_{1 \le \sabs{\li} <  K} \,\frac{1}{\sabs{\li}^2} \cdot \frac{1}{n^2 \, \norm{\li \cdot \omm}^2} \\
\le \, \Bigr( \frac{S}{\rho} \Bigl)^2 \, \sum_{1 \le \sabs{\li} <  K} \,\frac{1}{\sabs{\li}^2} \cdot \frac{\sabs{\li}^{2 A}}{n^2 \, \ka^2}
\less  \, \Bigr( \frac{S}{\rho} \Bigl)^2 \,  \frac{K^{2 A}}{n^2 \ka^2}
\end{align*}

We conclude:
$$\bnorm{ \frac{1}{n} \sum_{j = 0}^{n-1}  u (\x + j \omm) - \avg{u} }_{L^2(\T^2)} \ \le \ \frac{S}{\rho} \Bigr( \frac{1}{K^{1/2}} + \frac{K^A}{n \ka} \Bigl) \le \frac{S}{\rho} n^{- a}$$
for some positive constant $a$ that depends on $A$ and for $n$ large enough depending on $A$ and $\ka$.

Using Chebyshev's inequality, the above estimate implies:
\begin{equation}\label{apriori-avshifts}
\mbox{mes }[ \x \in \T^2 : \abs{ \frac{1}{n} \sum_{j = 0}^{n-1}  u (\x + j \omm)  \, -  \avg{u} }  > \frac{S}{\rho} \, n^{- a/3} ] \ < \  
 n^{-4 a/3}  
\end{equation}

This is not exactly what we wanted, since the size of the ``bad'' set above decays only polynomially fast in $n$, instead of exponentially fast. 

To boost this estimate, we will use Lemma 4.12 in J. Bourgain's monograph \cite{B}.  This result shows that a weaker a-priori estimate on a subharmonic function implies an upper bound on its BMO norm, which in turn leads, via John-Nirenberg inequality, to a stronger estimate on the function. We reproduce here a ``rescaled'' version of the estimate in \cite{B}, one that takes into account the width $\rho$ of subharmonicity. The reader may verify that this is indeed the correct rescaled version of the statement.

\begin{lemma}\label{boost}
Assume that $ u = u (\x) \colon \mathbb{T}^2 \to \mathbb{R} $ has a pluri-subharmonic extension $u (\z)$ on $\Aa_{\ro} = [ \sabs{ \Im z_1 } < \rho_1 ] \times   [ \sabs{ \Im z_2 } < \rho_2 ] $ such that $\displaystyle \sup_{\z \in \Aa_{\ro}} \, \abs{ u (\z) } \le B$.  Let $\rho = \min \{\rho_1, \rho_2 \}$.  If
\begin{equation}\label{weak}
 \mbox{ mes } [ \x \in \mathbb{T}^2 :  \abs{ u (\x) - \avg{u} }  > \epsilon _0 ] < \epsilon _1
 \end{equation}
then for an absolute constant $c > 0$, 
\begin{equation}\label{strong}
 \mbox{ mes } [ \x \in \mathbb{T}^2 :  \abs{  u (\x) - \avg{u} } > {\epsilon _0}^{1/4} ] 
 < e^{- c \bigl(  {\epsilon _0}^{1/4} + \sqrt{\frac{B}{\rho}}  \;   \frac{{\epsilon_1}^{1/4}} {{\epsilon_0}^{1/2}}    \bigr)^{- 1}}
 \end{equation}
\end{lemma}

\smallskip

We will apply this result to the average
$$u^\sharp (\x) := \frac{\rho}{S} \,  \frac{1}{n}\sum_{j = 0}^{n-1}  u (\x + j \omm)$$ 

Clearly $u^\sharp (\x)$ is pluri-subharmonic on the same strip $ \Aa_{\ro}$ as $u (\x)$, its upper bound on this strip is $B = \rho$ and its mean is $\displaystyle \bigr< u^\sharp \bigl> =  \frac{\rho}{S} \, \avg{u}$

Then \eqref{apriori-avshifts} implies
\begin{equation}\label{apriori-avg}
\mbox{mes }[ \x \in \T^2 \colon \abs{ u^\sharp (\x) -  \bigr< u^\sharp \bigl>  } > \ep_0 ] < \ep_1
\end{equation}
where $\ep_0 := n^{- a/3}$ and $\ep_1 := n^{-4 a/3}$ so  $\ep_1 \ll \ep_0$.

Applying Lemma~\ref{boost} and  performing the obvious calculations, from inequality \eqref{strong} we get
$$  \mbox{mes }[ \x \in \T^2 \colon \abs{ u^\sharp (\x) -  \bigr< u^\sharp \bigl>  } >  n^{- a/12} ] <  e^{-c \, n^{a/12}}$$
which then implies \eqref{shiftldt} for the multi-frequency shift $\mshift$.
\end{proof}

\bigskip

\section{{\L}ojasiewicz inequality for multivariable smooth functions}\label{lojasiewicz} 
To prove the large deviation estimate~\eqref{ldt-strategy} for a large enough initial scale $N_0$, we will need a quantitative description of the transversality condition~\eqref{TC}. 
More precisely, we will show that if a smooth function $v (\x)$ is not flat at any point as defined in \eqref{TC},  then the set  $[\x \colon v (\x) \approx E]$ of points where $v (\x)$ is almost  constant has small measure (and bounded complexity). 

Such an estimate is called a {\L}ojasiewicz type inequality and it is already available for non-constant analytic functions.  For such functions it can be derived using complex analysis methods from \cite{Levin}, namely lower bounds for the modulus of a holomorphic function on a disk (see Lemma 11.4 in \cite{GS1}). 

For non-analytic functions, the proof is more difficult.  Using Sard-type arguments, we have obtained a similar result for one-variable functions (see Lemma 5.3 in \cite{sK1}). For multivariable smooth functions, the argument is more technical and it involves a quantitative form of the implicit function theorem, also used in \cite{jC} and \cite{GS3}.  

We begin with a simple compactness argument that shows that in the TC~\eqref{TC} we can work with finitely many partial derivatives. 
\begin{lemma} \label{compactarg}
Assume $v (\x)$ is a smooth, $1$-periodic function on $\mathbb{R}^2$. Then $v (\x)$ satisfies the transversality condition (\ref{TC}) if and only if 
 \begin{equation}\label{TC'}
 \exists \, \m \in \N^2 \  \sabs{\m} \neq 0 \  \ \exists  c > 0 \ \colon  \  \forall \x \in \mathbb{T}^2  \  \max_{\genfrac{}{}{0cm}{}{\alfa \leq \m}{\sabs{\alpha} \neq 0}} \abs{ \partial^{\alfa} \, v (\x) } \geq c
\end{equation}
The constants $m, c$ in (\ref{TC'}) depend only on $v$.
\end{lemma}  
\begin{proof}
Clearly (\ref{TC'}) $ \Rightarrow $  (\ref{TC}). We prove the converse.  
The TC \eqref{TC} implies 
$$ \forall \, \x \in \T^2 \  \exists \, \m_{\x} \in \N^2 \  \sabs{\m_{\x}} \neq 0 \  \mbox{ such that } \  \abs{ \partial^{\m_{\x}} \, v \, (\x) } >  c_{\x}  > 0 $$

Then there are radii $  r_{\x} > 0 $ so that if  $\y$ is in the disk $D (\x, r_{\x})$ we have
$ \abs{ \partial^{\m_{\x}} \, v \, (\y) } \geq c_{\x}  > 0 $.
The family $ \{ D (\x, r_{\x})  \colon \x \in \T^2 \}$ covers $ \T^2$. Consider a finite subcover $\{  D (\x_1, r_{\x_1}), \ldots , D (\x_k, r_{\x_k}) \}$.
Let  $ \m  \in \N^2$ such that $ \m \ge \m_{\x_j} $ for all   $ 1 \leq j \leq k $ and  $\displaystyle c := \min_{1 \le j \le k} c_{\x_j} $.  Then (\ref{TC'}) follows.
\end{proof}
The following is a more precise form of the implicit function theorem (which was also used in \cite{GS3}).
\begin{lemma}\label{e-implicit} 
Let $f (\x)$ be a  $C^1$ function on a rectangle $\Rr = I \times J \subset [0, 1]^2$, let $J = [c, d]$ and $\displaystyle A := \max_{\x \in \Rr} | \dx f (\x)|$. Assume that
\begin{equation}\label{dx2>}
\min_{\x \in \Rr} | \dy f (\x)| =: \ep_0 > 0
\end{equation}
If  $f \aaa = 0$ for some point $ \aaa \in \Rr$, then there is an interval $I_0 = (a_1 - \ka, a_2 + \ka) \subset I$ and a $C^1$ function $\phi_0 (x_1) $ on $I_0$ such that:
\begin{eqnarray*}
\text{(i) }  & f (x_1, \phi_0 (x_1)) = 0 & \text{ for all }  x_1 \in I_0\\
\text{(ii) } &  | \dx \phi_0 (x_1)| \le A\,  \epm_0 & \\
\text{(iii) } & x_1 \in I_0 \text{ and } f \xx = 0 &  \implies x_2 = \phi_0 (x_1) 
\end{eqnarray*}
Moreover, the size $\ka$ of the domain of $\phi_0$ can be taken as large  as $\ka \sim \ep_0 \Am \cdot \min \{a_2-c, d-a_2 \}$.
\end{lemma}
\begin{proof}
From \eqref{dx2>}, since $\dy f (\x)$ is either positive on $\Rr$ or negative on $\Rr$ (in which case replace $f$ by $-f$), we may clearly assume that in fact: 
\begin{equation}\label{dx2p>}
\min_{\x \in \Rr}  \dy f (\x) =: \ep_0 > 0
\end{equation}

Moreover, note that for any fixed $x_1 \in I$, since $\dy f \xx \neq 0$, the equation $ f \xx =0$ has  a unique solution $x_2$.

Let $x_1 \in I_0$. Then
$$\abs{ f (x_1, a_2) } = \abs{ f (x_1, a_2) - f \aaa } = \abs{ \dx f(\xi, a_2) } \cdot \sabs{ x_1 - a_1 } \le A \ka$$ 

We have two possibilities.

\smallskip

$\blob$ \ $0 \le f (x_1, a_2) \le A \ka$. Then, if $a_2-t \in J$ we have:
$$f (x_1, a_2-t) - f (x_1,a_2) = \dy f (x_1, \xi) \cdot (-t)$$
$$f (x_1, a_2-t) = f (x_1,a_2) - t \cdot \dy f (x_1, \xi) \le A \ka - t \ep_0 = 0$$

$\text{ provided } t = A \epm_0 \ka$

\smallskip

For this choice of $t$, $a_2-t$ is indeed in $J$, because of the size $\ka_0$ of the interval $I_0$:
$\ t = A \epm_0 \ka \le A \epm_0 \ep_0 \Am (a_2-c) = a_2-c,  \ \text{ so } a_2-t \ge c$.

\smallskip

Therefore,
$$f (x_1, a_2-t) \le 0 \le f (x_1,a_2)$$
so there is a unique $x_2 =: \phi_0 (x_1) \in [a_2-t, a_2]$ such that $f(x_1, \phi_0 (x_1)) = 0$.

\medskip

$\blob$ \ $- A \ka  \le f (x_1, a_2) \le 0$. Then, if $a_2+t \in J$ we have:
$$f (x_1, a_2+t) - f (x_1, a_2) = \dy f (x_1, \xi) \cdot t$$
$$f (x_1, a_2+t) = f (x_1, a_2) + t \cdot \dy f (x_1, \xi) \ge - A \ka + t \ep_0 = 0$$

$\text{ provided } t = A \epm_0 \ka$

\smallskip

As before, for this choice of $t$, $a_2+t$ is  in $J$, because of the size $\ka$ of the interval $I_0$:
$\ t = A \epm_0 \ka \le A \epm_0 \ep_0 \Am (d-a_2) = d-a_2, \ \text{ so } a_2+t \le d$.

\smallskip

Therefore,
$$f (x_1, a_2) \le 0 \le f (x_1, a_2+t)$$
so there is a unique $x_2 =: \phi_0 (x_1) \in [a_2, a_2+t]$ such that $f(x_1, \phi_0 (x_1)) = 0$.

\medskip

We proved (i) and (iii). The fact that $\phi_0 (x_1)$ is $C^1$ follows from the standard implicit function theorem, while the estimate (ii) follows immediately from (i) using the chain's rule.

\end{proof} 

\smallskip

The following is a quantitative and global version of the previous lemma (see also Lemma 8.3 in \cite{GS3}). It says that under the same conditions as above, the points $\xx \in \Rr$ for which  $\abs{ f \xx } \le \ep$  are either in a narrow strip at the top or at the bottom of the rectangle $\Rr$, or  near the graphs of some functions $\phi_j (x_1)$, in other words $x_2 \approx \phi_j (x_1)$.

\smallskip

\begin{lemma}\label{q-implicit} 
Let  $f (\x)$ be a $C^1$ function  on a rectangle $\Rr = I \times J \subset [0, 1]^2$, where $| I | \sim \ka_0$. Let  $J = [c, d]$ and $\displaystyle A := \max_{\x \in \Rr} | \dx f (\x)|$. Assume that:
\begin{equation}\label{dx2big}
\min_{\x \in \Rr} \abs{ \dy f (\x) } =: \ep_0 > 0
\end{equation}

Let $\ep_1 >0$ be small enough, i.e. $\ep_1 < \frac{\ep_0 \ka_0}{4}$
and $\ka_1 \sim \ep_1 \Am$.

\smallskip

Then  there are  about $\ka_0 \kam_1$ sub-intervals $I_j$ of length $\ka_1$  covering $I$, and on each interval $I_j$ there is a $C^1$ function $\phi_j (x_1) $ such that:
\begin{eqnarray*}
\text{(i) }  & f (x_1, \phi_j (x_1)) = 0 & \text{for all }  x_1 \in I_j  \\
\text{(ii) } &  \abs{ \dx \phi_j (x_1) } \le A\,  \epm_0 &   \\
\text{(iii) } & [ \xx \in \Rr \colon  \abs{ f \xx } < \ep_1 ]    \subset  & \Rr^t \cup \Rr^b \cup (\cup_{j} \ \Ss_j) \nonumber
\end{eqnarray*}
where
\begin{eqnarray*}
\Rr^t & := & I \times [d-2 \ep_1 \epm_0, d]\\
\Rr^b & := & I \times [c, c+2 \ep_1 \epm_0]\\
\Ss_j & := & [ \xx \colon x_1 \in I_j, \abs{ x_2 - \phi_j (x_1) } < \ep_1 \epm_0 ]
\end{eqnarray*}
\end{lemma}

\begin{figure}[h]
\includegraphics[scale=.85]{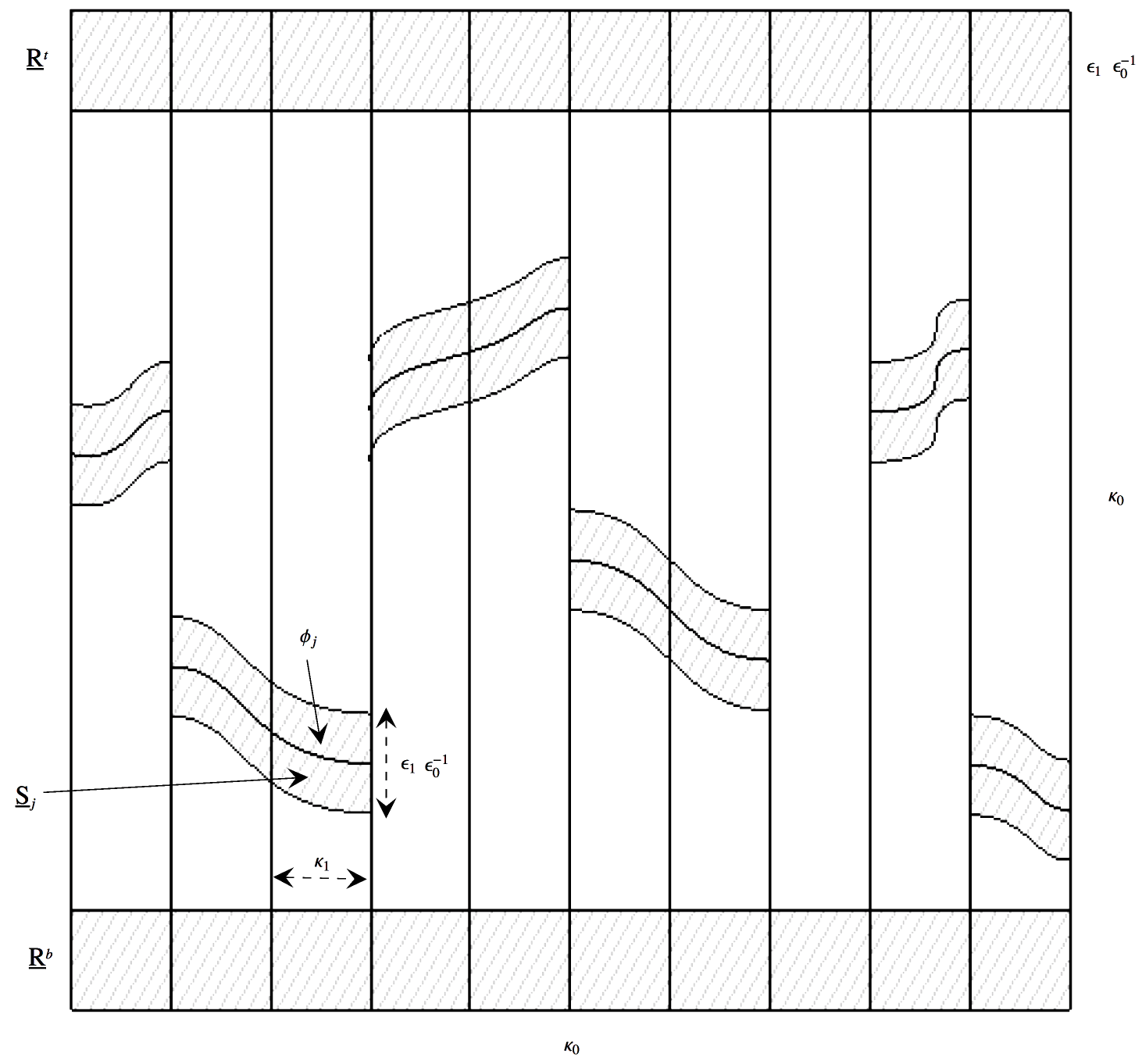}
\end{figure}

\begin{proof}
Divide the interval $I$, whose length is $\sim \ka_0$ into $\sim \ka_0\kam_1$ sub-intervals $I_j$ of length $\ka_1$ each. 

If $\abs{ f \xx } \ge \ep_1$ for all $\xx \in I_j \times [c + 2 \ep_1 \epm_0, d -  2 \ep_1 \epm_0]$, then we are done with the interval $I_j$.

Otherwise, assume $\abs{ f \aaa } < \ep_1$ for some $a_1 \in I_j$ and $a_2 \in [c + 2 \ep_1 \epm_0, d -  2 \ep_1 \epm_0]$.

We may assume $0 \le f \aaa \le \ep_1$, the other case being treated similarly. Then if $a_2 - t \in J$ we have:
$$f (a_1, a_2 - t) - f \aaa = \dy f (a_1, \xi) \cdot ( - t) \ \text{ for some } \xi \in (a_2 - t, a_2)$$
$$ f (a_1, a_2 - t) =  f \aaa - t \cdot  \dy f (a_1, \xi) \le \ep_1 - \ep_0 t = 0$$
provided $t = \ep_1 \epm_0$. Since $a_2 \ge c+2 \ep_1 \epm_0$, for this $t$ we have $a_2-t \in J$.

\smallskip

We then have $f (a_1, a_2 - t) \le 0 \le f \aaa$, so $f (a_1, a_2^{*}) = 0$ for some $a_2^{*} \in [a_2-t, a_2]$.

We can use Lemma \ref{e-implicit} around the point $(a_1, a_2^{*})$. The interval we get has length at least $\ep_0 \Am \cdot \min \{a_2-c, d-a_2 \} > \ep_0 \Am \cdot  2 \ep_1 \epm_0 = 2 \ep_1 \Am > 2 \ka_1$, so it contains $I_j$, whose length is $\sim \ka_1$. We have a $C^1$ function $\phi_j$ on $I_j$ such that $| \dx \phi_j | \le A \epm_0$ and
$$  x_1 \in I_j \text{ and }  f \xx = 0 \iff x_2 = \phi_j (x_1)$$

Now let $\xx \in \Rr$ such that $\abs{ f \xx } < \ep_1$. Then either $\xx \in \Rr^t \cup \Rr^b$ or $\xx \in I_j \times [c + 2 \ep_1 \epm_0, d -  2 \ep_1 \epm_0]$ for some $j$, in which case:
$$\ep_1 > \abs{ f \xx } = | f \xx - f (x_1, \phi_j (x_1)) | = $$
$$= | \dy f (x_1, \xi) | \cdot | x_2 - \phi_j (x_1) | \ge \ep_0 \cdot   | x_2 - \phi_j (x_1) |$$
from which we conclude that $| x_2 - \phi_j (x_1) | < \ep_1 \epm_0$.

\end{proof}

We have shown that the points $\x = \xx \in \Rr$ for which $| f (\x) | < \ep_1$ are within $\sim \ep_1$ from the graphs of some functions $\phi_j (x_1)$ that have bounded slopes and are defined on small intervals $I_j$. This shows that the 'bad'  set $[ \x \in \Rr \colon | f (\x) | < \ep_1 ]$ can be covered by small rectangles instead of $\ep_1$-neighborhoods of curves, and we have control on the size of these rectangles and on their number. In turn, the 'good' set $[ \x \in \Rr \colon | f (\x) | \ge \ep_1 ]$ can be covered by a comparable number of rectangles, which can be further chopped down into squares, to preserve the symmetry between the two variables. This is the content of the following lemma. 

\begin{lemma}\label{L-ind}
Given a $C^2$ function $f (\x)$ on a square $\Rr_0 = I_0 \times J_0 \subset [0, 1]^2$, where $| I_0 |, |J_0| \sim \ka_0$.  
Denote $\displaystyle A := \max_{ |\alfa| \le 2} \ \max_{\x \in \Rr} | \partial^{\alfa} f (\x)|$. 
Assume that:
\begin{equation}\label{dx2big2}
\min_{\x \in \Rr_0} | \dy f (\x)| =: \ep_0 > 0 \    \text{ or }  \ \min_{\x \in \Rr_0} | \dx f (\x)| =: \ep_0 > 0 
\end{equation}

Let $\ep_1 >0$ be small enough, i.e. $\ep_1 < \frac{\ep_0 \ka_0}{4}$
and $\ka_1 \sim \ep_1 \Am$.

Then there is a set $\Bb_1 \subset \Rr_0$, with 
\begin{equation}\label{bad1}
 \text{ mes } [ \Bb_1 ] \less \ka_0 \, \ep_1 \, \epm_0 
 \end{equation} 
 such that $ \Rr_0 \setminus \Bb_1 $ is a 
union of about $(\ka_0 \, \kam_1)^2$ squares, where each such square has the form  $\Rr_1 = I_1 \times J_1$, with  $\abs{I_1}, \abs{J_1} \sim \ka_1$.

For each of  these squares we have:
\begin{equation}\label{f>ep1}
 \min_{\x \in \Rr_1} \ \abs{ f (\x) } \ge \ep_1
 \end{equation}

\end{lemma}

\begin{proof}
We will use Lemma \ref{q-implicit}.  $I_0$ is covered by about $\ka_0 \kam_1$ subintervals $I_j$ of length $\ka_1$.  Consider one such subinterval. There is a $C^1$ function $\phi_j (x_1) $ on $I_j$ such that $ | \dx \phi_j (x_1)| \le A\,  \epm_0$ and 
$$ [ \xx \in I_j \times J_0  \colon  \abs{ f \xx } < \ep_1 ]    \subset   \Rr_j^t \cup \Rr_j^b  \cup   \Ss_j $$
where if $J_0 = [c_0, d_0]$ then 
\begin{eqnarray*}
\Rr_j^t & := & I_j  \times [d_0-2 \ep_1 \epm_0, \, d_0]\\
\Rr_j^b & := & I_j  \times [c_0, \, c_0+2 \ep_1 \epm_0]\\
\Ss_j & := & [ \xx \colon x_1 \in I_j, |x_2 - \phi_j (x_1) | < \ep_1 \epm_0 ]
\end{eqnarray*}
Then
$$\Ss_j \subset I_j \times [ \min_{x_1 \in I_j} \, \phi_j (x_1) - \ep_1 \epm_0, \ \max_{x_1 \in I_j} \, \phi_j (x_1) + \ep_1 \epm_0] =: I_j \times K_j^m =: \Rr_j^m$$
For any $x_1, x_1^\prime \in I_j$ we have 
$$\abs{ \phi_j (x_1) - \phi_j (x_1^\prime)} \less A \epm_0 \cdot \abs{x_1 - x_1^\prime} \le A \epm_0 \ka_1 \sim \ep_1 \epm_0 $$
which shows that 
$$\abs{ K_j^m } \less  \ep_1 \epm_0 $$

We have shown that $ [ \xx \in I_j \times J_0  \colon  \abs{ f \xx } < \ep_1 ] $ is covered by three rectangles: $\Rr_j^t$, $\Rr_j^b$, $\Rr_j^m$, each of the form $I_j \times K_j$ where $ \abs{I_j} \sim \ka_1$, $\abs{K_j} \sim \ep_1 \epm_0$.

\smallskip

Summing over $j \less \ka_0 \kam_1$, we get that the set $ [ \x \in \Rr_0 \colon \abs{ f (\x) } < \ep_1 ] $ is contained in the union $\Bb_1$ of about $ \ka_0 \kam_1$ rectangles of size $ \ka_1 \times \ep_1 \epm_0$. Then 

$$\text{ mes } [ \Bb_1 ] \less  \ka_0 \kam_1 \cdot  \ka_1 \cdot  \ep_1 \epm_0 = \ka_0 \ep_1 \epm_0$$
which proves \eqref{bad1}.

\smallskip

The complement of this set,  $ \Rr_0 \setminus \Bb_1 $, consists of about the same number $\ka_0 \kam_1$ of rectangles - this was the reason for switching from $\ep_1$-neighborhoods of curves to rectangles. 
Each of these rectangles has the form $I_j \times L_j$, where $\abs{ I_j } \sim \ka_1$ and  $\abs{ L_j } \sim \ka_0 - O(\ep_1 \epm_0) \sim \ka_0 \gg \ka_1$. Divide each of these vertical rectangles into about $\ka_0 \kam_1$ squares of size $\ka_1 \times \ka_1$ each.

We conclude that $ \Rr_0 \setminus \Bb_1 $ is covered by about $(\ka_0 \kam_1)^2$ squares of the form $\Rr_1 = I_1 \times J_1$, where the size of each square is $\abs{I_1}, \abs{J_1} \sim \ka_1$.

\end{proof}

\medskip

We now have all the ingredients for proving the following {\L}ojasiewicz type inequality. 

\begin{theorem}\label{Loj}
Assume that $v (\x)$ is a smooth function on $[0,1]^2 $ satisfying the  transversality condition (\ref{TC}).
Then for every $\ep > 0$
\begin{equation}\label{loj} 
\sup_{E \in \mathbb{R}} \mbox{ mes } [ \x \in [0,1]^2 : \, | v (\x) - E | < \ep ] < C \cdot \ep^{b}
\end{equation}
where $C, b > 0 $ depend only on $v$. 
 \end{theorem}
 
 \begin{proof}
 Using Lemma \ref{compactarg}, 
 $$ \exists \, \m = \mm \in \N^2 \  \sabs{\m} \neq 0 \  \ \exists  c > 0 \ \colon  \  \forall \x \in \mathbb{T}^2  \  \max_{\genfrac{}{}{0cm}{}{\alfa \leq \m}{\sabs{\alpha} \neq 0}}  \abs{ \partial^{\alfa} \, v (\x) } \geq c$$
 Let 
 $$A:= \max_{\alfa \le (m_1+1, m_2+1)} \ \max_{\x \in [0,1]^2} \,  \abs{\partial^{\alfa} \, v (\x) } $$
 
We may of course assume that  $\abs{ E } \le 2 A$, otherwise there is nothing to prove.

All the constants in the estimates that follow will depend only on $\sabs{\m}, c, A$ (so in particular only on $v$).

\smallskip

Partition $[0,1]^2$ into about $(\frac{2 A}{c})^2$ squares of the form $\Rr = I \times J$ of size $\abs{I}, \abs{J} \sim \frac{c}{2A}$.

Let $\Rr$ be such a square. Then either
$\sabs {v (\x) } \ge \ep \ \text{ for all } \x \in \Rr$, 
in which case we are done with this square, or for some $\ab = \aaa \in \Rr$ we have $\sabs{ v (\ab) } < \ep$. But then for one of the partial derivatives $\alfa \le \m$,  $\sabs{\alpha} \neq 0$, we have $\abs{ \partial^{\alfa} \, v (\ab)  } \ge c$.

\smallskip

Assume for simplicity that $\abs{ \partial^{\m} \, v (\ab)  } \ge c$, which is the worst case scenario.

If $\x \in \Rr$, then $\norm{ \x - \ab}_{\infty} := \max \{ \sabs{x_1 - a_1}, \sabs{x_2-a_2} \} \le \frac{c}{2A}$.  
 
Then 
$$ \abs{ \partial^{\m} \, v (\x) -  \partial^{\m} \, v (\ab)}  \less \max_{\y \in \Rr } \, \abs{\nabla \partial^{\m} \, v (\y) } \cdot \norm{ \x - \ab}_{\infty}  \le
A \cdot \frac{c}{2A} = \frac{c}{2}$$
 
It follows that 
$$  \min_{\x \in \Rr }  \abs{ \partial^{\m} \, v (\x) } \more \frac{c}{2}$$

We will use Lemma \ref{L-ind} $\sabs{\m} =: m$ times. 

\medskip

$\blob$ \textbf{Step 1.} Let $$f_1 (\x) := \partial^{(m_1, m_2-1)} \, v (\x)$$
Then 
$$\min_{\x \in \Rr} \abs{\dy f_1 (\x) } =  \min_{\x \in \Rr }  \abs{ \partial^{\m} \, v (\x) } \more c$$

We apply Lemma \ref{L-ind} to the function $f_1$ with the following data: 
$$\Rr_0 = \Rr, \ \ka_0 = \frac{c}{2A}, \ \ep_0 \sim c, \ \ep_1 < \frac{\ep_0 \ka_0}{4}, \ \ka_1 \sim \ep_1 \Am$$
where $\ep_1$ will be chosen later. 

We get a set $\Bb_1^{\flat}:= \Bb_1$, $\text{ mes } [ \Bb_1^\flat ] \less \ka_0 \ep_1 \epm_0 < \ka_0^2 A \cdot \ep_1 \epmt_0$ such that $\Rr_0 \setminus  \Bb_1^\flat$ is a union of about 
$(\ka_0 \kam_1)^2$ squares of the form $\Rr_1 = I_1 \times J_1$, of size $\abs{I_1}, \abs{J_1} \sim \ka_1$. For each of these squares we have:
$$\min_{\x \in \Rr_1} \abs{f_1 (\x) }  \ge \ep_1$$
which means:
$$\min_{\x \in \Rr_1} \abs{\ \partial^{(m_1, m_2-1)} \, v (\x) } \ge \ep_1$$

\medskip

$\blob$ \textbf{Step 2.} Pick any of the squares $\Rr_1 = I_1 \times J_1$ from the previous step and consider say 
$$f_2 (\x) := \partial^{(m_1-1, m_2-1)} \, v (\x)$$
Then
$$\min_{\x \in \Rr_1} \abs{\dx f_2 (\x) } =  \min_{\x \in \Rr_1 }  \abs{ \partial^{(m_1, m_2-1)} \, v (\x) } \ge \ep_1$$

Apply Lemma \ref{L-ind} to the function $f_2$ with the following data: 
\begin{center}
$\Rr_1$, $\ka_1$, $\ep_1$ from Step 1, \ $\ep_2 < \frac{\ep_1 \ka_1}{4}$, \ $\ka_2 \sim \ep_2 \Am$
\end{center}
where $\ep_2$ will be chosen later. 
 
We get a set $\Bb_2$, $\text{ mes } [ \Bb_2 ] \less \ka_1 \ep_2 \epm_1$ such that $\Rr_1 \setminus  \Bb_2$ is a union of about $(\ka_1 \kam_2)^2$ squares of the form $\Rr_2 = I_2 \times J_2$, of size $\abs{I_2}, \abs{J_2} \sim \ka_2$. For each of these squares we have:
$$\min_{\x \in \Rr_2} \abs{f_2 (\x) }  \ge \ep_2$$
which means:
$$\min_{\x \in \Rr_2} \abs{\ \partial^{(m_1-1, m_2-1)} \, v (\x) } \ge \ep_2$$ 

\smallskip

If we do this for each of the $\sim (\ka_0 \kam_1)^2$ squares resulting from Step 1, and if we put together all the `bad' sets $\Bb_2$ corresponding to each of these squares, we conclude the following.

\smallskip

There is a set $\Bb_2^{\flat} \subset \Rr$ such that:
$$\text{ mes } [ \Bb_2^{\flat} ] \less \ka_1 \ep_2 \epm_1 \cdot (\ka_0 \kam_1)^2 = \ka_0 \ep_2 \epm_1 \kam_1 \sim \ka_0^2 A \cdot \ep_2 \epmt_1$$

Hence the total measure of the `bad' set in Step 2 is:
$$\text{ mes } [ \Bb_2^{\flat} ] \less  \ka_0^2 A \cdot \ep_2 \epmt_1$$

Moreover, $\Rr \setminus (\Bb_1^{\flat} \cup \Bb_2^{\flat}) $ is covered by squares of the form $\Rr_2 = I_2 \times J_2$, of size $\abs{I_2}, \abs{J_2} \sim \ka_2$.

The total number of such squares is about $$(\ka_1 \kam_2)^2 \cdot (\ka_0 \kam_1)^2 = (\ka_0 \kam_2)^2$$

On each of these squares we have:

$$\min_{\x \in \Rr_2} \abs{\ \partial^{(m_1-1, m_2-1)} \, v (\x) } \ge \ep_2$$ 

\bigskip

It is clear how this procedure continues. Perform it for $m-1$ steps. We will get sets $\Bb_1^{\flat}, \dotsc,  \Bb_{m-1}^{\flat} $ such that 
$\Rr \setminus (\Bb_1^{\flat} \cup \dotsc \cup \Bb_{m-1}^{\flat}) $ consists of about  $(\ka_0 \kam_{m-1})^2$ squares of the form $\Rr_{m-1} = I_{m-1}  \times J_{m-1}$, of size $\abs{I_{m-1}}, \abs{J_{m-1}} \sim \ka_{m-1}$. 
On each of these squares we have 
$$\min_{\x \in \Rr_{m-1}} \abs{\ \dy \, v (\x) } \ge \ep_{m-1} \ \text{ or } \ \min_{\x \in \Rr_{m-1}} \abs{\ \dx \, v (\x) } \ge \ep_{m-1}$$

$\blob$ \textbf{Step m.} Assume the former inequality above and apply Lemma \ref{L-ind} one more time.
Let $$f_m (\x) := v (\x) - E$$
for some fixed energy $E$ with $\abs{E} \le 2A$ (the estimates will not depend on $E$).
Then for each of the squares $\Rr_{m-1}$ from the previous step we have:
$$  \min_{\x \in \Rr_{m-1}} \abs{ \dy \, f_{m} (\x) }  =  \min_{\x \in \Rr_{m-1}} \abs{\dy \, v (\x) } \ge \ep_{m-1}$$

Apply Lemma \ref{L-ind} to the function $f_m$ with the following data: 
\begin{center}
$\Rr_{m-1}$, $\ka_{m-1}$, $\ep_{m-1}$ from the previous step, \ $\ep_m < \frac{\ep_{m-1}  \ka_{m-1}}{4}$, \ $\ka_m \sim \ep_m \Am$
\end{center}
where $\ep_m$ will be chosen later. 
 
We get a set $\Bb_m$, $\text{ mes } [ \Bb_m ] \less \ka_{m-1} \ep_m \epm_{m-1}$ such that $\Rr_{m-1} \setminus  \Bb_m$ is a union of about $(\ka_{m-1} \kam_m)^2$ squares of the form $\Rr_m = I_m \times J_m$, of size $\abs{I_m}, \abs{J_m} \sim \ka_m$. For each of these squares we have:
$$\min_{\x \in \Rr_m} \abs{f_m (\x) }  \ge \ep_m$$
which means:
$$\min_{\x \in \Rr_m} \abs{ v (\x) - E } \ge \ep_m$$ 

\smallskip

If we do this for each of the $\sim (\ka_0 \kam_{m-1})^2$ squares resulting from the previous step, and if we put together all the corresponding `bad' sets, we conclude.

\smallskip

There is a set $\Bb_m^{\flat} \subset \Rr$ such that:
$$\text{ mes } [ \Bb_m^{\flat} ] \less \ka_{m-1} \ep_m \epm_{m-1} \cdot (\ka_0 \kam_{m-1})^2 = \ka_0 \ep_m \epm_{m-1} \kam_{m-1} \sim \ka_0^2 A \cdot \ep_m \epmt_{m-1}$$

Hence the total measure of the `bad' set in Step m is:
$$\text{ mes } [ \Bb_m^{\flat} ] \less  \ka_0^2 A \cdot \ep_m \epmt_{m-1}$$

Moreover, $\Rr \setminus (\Bb_1^{\flat} \cup \Bb_2^{\flat} \dotsc \cup \Bb_m^{\flat} ) $ is covered by squares of the form $\Rr_m = I_m \times J_m$, of size $\abs{I_m}, \abs{J_m} \sim \ka_m$.

The total number of such squares is about $$(\ka_{m-1} \kam_m)^2 \cdot (\ka_0 \kam_{m-1})^2 = (\ka_0 \kam_m)^2$$

On each of these squares we have:

$$\min_{\x \in \Rr_m} \abs{ v (\x) - E } \ge \ep_m$$ 

Therefore, the total measure of the bad set from all steps is:
\begin{equation}\label{total-bad}
\text{ mes } [ \Bb_1^{\flat} \cup \Bb_2^{\flat} \dotsc \cup \Bb_m^{\flat}  ] \less \ka_0^2 A \cdot [ \ep_1 \epmt_0 +  \ep_2 \epmt_1 + \dotsc  \ep_m \epmt_{m-1} ] 
\end{equation}

We choose $$\ep_j := \ep^{1/3^{m-j}} \ \text{ for } 1 \le j \le m$$

If $\ep < \ep^{*}  (c, m)$, then $\ep_0 \sim c > \ep^{1/3^m}$,  $\epmt_0 \sim c^{-2} < \epmt$, so there is no harm in also putting (for simplicity)  $\ep_0 = \ep^{1/3^m}$.

It is a simple calculation to see that for any $\ep < \ep^{*} (m, A)$, we have $\ep_{j+1} < \frac{\ep_j \ka_j}{4}$ for all $j = \overline{0 \ldots m-1}$, which allows our inductive process to work.

Note that $\ep_j^3 = \ep_{j+1}$ so $\ep_{j+1} \epmt_j = \ep_j$. This implies:
$$ \ep_1 \epmt_0 +  \ep_2 \epmt_1 + \dotsc  \ep_m \epmt_{m-1} = \ep_0 + \ep_1 + \dotsc + \ep_{m-1} \le m \ep_0 = m \cdot \ep^{1/3^m}$$

From \eqref{total-bad} it follows that the total measure of the bad set inside the square  $\Rr$ is:
$$\text{ mes } [ \Bb_1^{\flat} \cup \Bb_2^{\flat} \dotsc \cup \Bb_m^{\flat}  ] \less \ka_0^2 \, A \, m \cdot \ep^{1/3^m}$$
There are about $(\frac{2A}{c})^2 = \kamt_0$ such squares. 

We conclude that outside a bad set $\Bb$, $\text{ mes } [ \Bb ] < A \, m \cdot e^{1/3^m} $, we have  $\abs{ v (\x) - E } \ge \ep$, which proves \eqref{Loj} with $C \sim A \, m$ and $b = \frac{1}{3^m}$.

 \end{proof}

\begin{remark} \rm{The exponent $b$ in \eqref{loj}  is related to the {\L}ojasiewicz exponent of the function $v$ (see \cite{jK}, \cite{DK}).  Determining the optimal exponent in such an inequality is an interesting problem in itself, and has been studied extensively for polynomials and analytic functions. It is clear that for a polynomial, the {\L}ojasiewicz exponent should be related to its degree $d$, and it is in fact shown to be $O (\frac{1}{d^2})$ with explicit underlying  constants (see \cite{DK}, \cite{jK}).  The proof of the {\L}ojasiewicz inequality for analytic functions in \cite{GS1} (see Lemma 11.4 there) does not provide an explicit value for the exponent, but Theorem 4 in \cite{PSS} provides a scheme for computing it via the Newton distance of $v$. 
 
 In our proof for smooth, transversal  functions, we obtain the exponent $\frac{1}{3^m}$, where $m$ is the maximum number of partial derivatives needed for transversality. If $v$ were a polynomial of degree $d$, then $m$ would be $d$, which shows that our estimate is very wasteful (we have obtained a better estimate, $O (\frac{1}{m})$, for one-variable functions, see Lemma 5.3 in \cite{sK1}). This, however,  seems to be the only such estimate available now for non-analytic functions of two variables. 
 
 A similar argument can be made for functions of more than two variables, so \eqref{loj} will hold for such functions as well.
}  
\end{remark}    

\bigskip

\section{Large deviation theorem, the proof of main results}\label{ldt_proofs} 
Using induction on the scale $N$, we will prove the large deviation estimate~\eqref{ldt-strategy} for the logarithmic average of  transfer matrices:
$$\mbox{mes } [ \x \in \T^2 \colon \abs{ \frac{1}{N} \log \norm{ M_{N} (\x, E) } - L_{N} (E) } > N^{-\tau} ] <  e^{-N^\sigma}$$
as well as a lower bound on the mean of these quantities:  
$$L_N (E) \ge \gamma_N \, \log \sabs{\la}$$

The base step of the induction uses the quantitative description~\eqref{loj} of the transversality condition~\eqref{TC} on the potential function, and the large size of the coupling constant. The inductive step uses only the regularity of the potential function via Lemma~\ref{lemma1}, which provides a good approximation of these logarithmic averages by pluri-subharmonic functions.  

\smallskip


\begin{lemma}{(Base step of the induction)} \label{base-step}
Assume  that $ v (\x) $ is smooth and satisfies the transversality condition (\ref{TC}). Then given any constant $C > 0$, there are positive constants $ \lambda_1$ and $B$ which depend on $v$ and $C$, such that
for any scale $N_0$,  for any $ \lambda$ subject to
$ \sabs{\la} \ge \max \{ \la _1, N_0^B \} $ and for any $E \in \R$ we have:
\begin{equation}\label{step1}
\mbox{ mes } [  \x \in \mathbb{T}^2 \colon 
\, \abs{  \frac{1}{N_0} \log \norm{  M_{N_0} (\x, \lambda, E) } - L_{N_0} (\lambda, E) }
> \frac{1}{20} \, S(\lambda) \, ]  < N_0^{- C}
\end{equation}

Furthermore, for these $\lambda$, $N_0$ and for all $E$ we have:
\begin{align}
 L_{N_0} (\lambda, E) \geq \frac{1}{2} \, S(\lambda)  \label{step10}  \\
 L_{N_0} (\lambda, E) -  L_{2 N_0} (\lambda, E) \leq  \frac{1}{80} \, S(\lambda)  \label{step100}
\end{align}
\end{lemma}
\begin{proof} The proof of this result is similar to the analytic potential function case.  That is because the only fact about analyticity needed here is the {\L}ojasiewicz inequality \eqref{loj}, which holds for any non-constant analytic functions, and which we have established in section~\ref{lojasiewicz} for smooth functions satisfying the transversality condition \eqref{TC}. We will then omit the proof, but the reader is referred  to the proof of Lemma 2.10 in \cite{BGS} for details. 
\end{proof}


We will now explain the idea of the proof of the inductive step. 

If at scale $N_0$ we apply  the almost invariance property (\ref{mshift}) $n$ times and then average, we get:
\begin{equation}\label{mshift2}
 \abs{ L_{N_0} (\x)  -  \frac{1}{n} \sum_{j = 0}^{n-1}  L_{N_0} (\shift^j \x) }  \less  \frac{n S(\la)}{N_0} 
\end{equation}
so using the approximation \eqref{aproxu}, we also get:
\begin{equation}\label{mshift3}
 \abs{ u_{N_0} (\x)  -  \frac{1}{n} \sum_{j = 0}^{n-1}  u_{N_0} (\shift^j \x) }  \less  \frac{n S(\la)}{N_0} 
\end{equation}
 
To have a decay above, we need to take a smaller number of shifts $n \ll N_0$. 

Apply the estimate  (\ref{shiftldt}) on averages of shifts of pluri-subharmonic functions to $u_{N_0} (x)$ and get:
\begin{equation}\label{shiftldt2} 
\mbox{mes }[ \x \in \T^2 : | \, \frac{1}{n} \sum_{j = 0}^{n-1}  u_{N_0} (\shift^j \x)  \, -  \avg{u_{N_0}}   | > \frac{S}{\rho_{N_0}} \, n^{- \tau_0} ] < 
e^{- n^{\sigma_0}}
\end{equation}

We may combine \eqref{mshift3}, \eqref{shiftldt2} to directly obtain a large deviation estimate for $u_{N_0} (x)$ and then, via the approximations  \eqref{aproxu}, \eqref{aprox<u>} to obtain the LDT for $L_{N_0} (x)$, \emph{only} when the deviation $\frac{S}{\rho_{N_0}} n^{- \tau_0}  \ll 1$. In other words, this approach works only when the scaling factor $\frac{S}{\rho_{N_0}}$ is not too large to cancel the decay $n^{- \tau_0}$. This is the case of the single or multi-frequency shift model with \emph{analytic} potential (see \cite{B}, \cite{BG}) where $\frac{S}{\rho_{N_0}} = \frac{S}{\rho}$ is just a constant depending on the potential function  $v$. This approach also works for the \emph{single}-frequency model with potential function in a Gevrey class of order  $s<2$, since in this case sharper estimates than (\ref{shiftldt2}) are available for averages of shifts of single-variable subharmonic functions  (see \cite{sK1}). This approach fails for the skew-shift model (whether the potential function is analytic or Gevrey) and also for the multi-frequency model with Gevrey potential function, because the size  $\rho_{N_0}$ of the subharmonic extension depends on the scale $N_0$. 

Therefore,  in order to beat the scaling factor $\frac{S}{\rho_{N_0}}$ when applying the estimate  (\ref{shiftldt2})  to a transfer matrix substitute $u_{N_0} (\x) $ at scale $N_0$, we need to consider a large number of shifts  $n \gg N_0$. The averages of shifts thus obtained will be close to the mean $\avg{u_{N_0}}$. Moreover, we will get:
 \begin{align*}
 L_{N_0}\overset{(1)}\approx \, \avg{u_{N_0} } \, \overset{(2)} \approx  \frac{1}{n} \sum_{j = 0}^{n-1}  u_{N_0} (\shift^j \x) \\
 \overset{(3)} \approx  \, \frac{1}{n} \sum_{j = 0}^{n-1}  \frac{1}{N_0} \log \norm{ M_{N_0} (\shift^j \x) } \overset{(4)} \approx \, 
  \frac{1}{n N_0} \log \norm{ M_{n N_0} (\x) } 
  \end{align*}
  
 The first approximation above is just (\ref{aprox<u>}). The second is exactly (\ref{shiftldt2}). 
 The third is due to  (\ref{aproxu}). The last approximation above essentially says that: 
$$\prod_{j=0}^{n-1}  \norm{ M_{N_0} (\shift^j \x) } \approx \norm{ \prod_{j=0}^{n-1}   M_{N_0} (\shift^j \x) } \approx \norm{   M_{n N_0} (\x) } $$
or in other words, that the product of the norms of certain transfer matrices is approximately equal to the norm of the product of these matrices, the latter giving us the transfer matrix at the larger scale $n N_0$.

If these heuristics were true, then for $n \gg N_0 $ we would get 
$$ L_{N_0} \approx  \frac{1}{n N_0} \log \norm{ M_{n N_0} (\x) } $$
which would establish the large deviation estimate for transfer matrices at a larger scale $n N_0$.   
  
The avalanche principle, which is a deterministic result, describes how estimates on the norms of individual (and of products of two consecutive) $SL_2 (\mathbb{R})$ matrices can lead to estimates on the norm of the product of all matrices (see \cite{GS1}, \cite{B}), thus providing the basis for establishing the above heuristics. It requires a uniform lower bound on the norms of individual matrices in the product, as well as knowing that the norm of the product of any two consecutive matrices is comparable to the product of their norms. 


\smallskip

The following lemma provides the inductive step in proving the LDT for an increasing sequence of scales $N$. It also provides the inductive step in proving the positivity and continuity of the Lyapunov  exponent. The  proof of this lemma is based on the heuristics described above, and combines the averages of shifts estimate (\ref{shiftldt}), the almost invariance property (\ref{mshift}) and the avalanche principle (see Proposition 2.2 in \cite{GS1}).

Before stating the lemma let us describe the various parameters and constants that will appear.

\smallskip

\noindent \textbf{List of constants and parameters:} 

\smallskip

$\dd$ $s > 1$ is the order of the Gevrey class.  

$\dd$ $\delta = 2 (s-1)$ refers to the size ($\approx N^{-\delta}$) of the holomorphic extensions of the transfer matrix substitutes.

$\dd$ $ D := 2 \delta + 8 $, $ A := \max \{\frac{2 \, (\delta+1)}{\tau_0}, \, 2 \} $ are some well chosen powers of the scale $N$, $\tau_0$ is the exponent from (\ref{shiftldt}). 

$\dd$  $\gamma > \frac{1}{4}$ is a fixed number.

\smallskip

Note that all these constants are either universal or depend on the order $s$ of the Gevrey class.

$\dd$ $\lambda $, $E$ are fixed parameters such that $ \abs{E}  \le \abs{\la}  B + 2 $, and $\displaystyle  B := \sup_{ \x \in \mathbb{T}} \abs{v (\x)} $. 

$\dd$ The transformation $\shift = \skews$ where $ \om \in DC_{\ka}$ or $\shift = \mshift$ where $ \omm \in DC_{\ka}$ for some $ \kappa > 0 $.

$\dd$ $N_{0 0} = N_{0 0} (s, \kappa, B)$ is a sufficiently large integer, such that the asymptotic behavior of various powers and exponentials applies to $N_{0 0}$ and such that (\ref{shiftldt}) holds for $N_{0 0}$  shifts.

\smallskip

\begin{lemma}{(The inductive step)}\label{ind} 
Consider two scales $N_0$ and $N$ such that $N_0 \geq N_{0 0}$,  (\ref{aproxu}) holds at scale $N_0$, that is:
\begin{equation}\label{scale2}
N_0 \geq S (\lambda) \hspace{.1in} \Leftrightarrow \hspace{.1in} 
| \lambda | \leq e^{ N_0}
\end{equation}
and 
\begin{equation}\label{scale1} 
N_{0}^{A} \leq N \leq  e^{N_0} 
\end{equation}

\smallskip

Assume that a weak LDT holds at scales $N_0$ and $2 N_0$:
 \begin{equation}\label{indhyp1}
\mbox{mes } [ \x \in \mathbb{T}^2 : \abs{ \frac{1}{N_0} \log \norm{ M_{N_0} (\x, \lambda, E ) } - L_{N_0} (\lambda, E ) } >  \frac{\gamma}{10} S (\lambda) ]  <  N^{- D}
\end{equation}
\begin{equation}\label{indhyp2}
 \mbox{mes } [ \x \in \mathbb{T}^2 : \abs{ \frac{1}{2 N_0} \log \norm{ M_{2 N_0} (\x, \lambda, E ) } - L_{2 N_0} (\lambda, E ) }  >   \frac{\gamma}{10}S (\lambda) ]  <  N^{- D}
\end{equation}
and that the means $L_{N_0}$,  $L_{2 N_0}$ have a lower bound and are close to each other:
\begin{eqnarray}
 L_{ N_0} (\lambda, E ), \, L_{2 N_0} (\lambda, E ) \geq & \gamma S(\lambda) \label{indhyp3} \\
 L_{ N_0} (\lambda, E ) -  L_{2 N_0} (\lambda, E ) \leq  & \frac{\gamma}{40}S (\lambda) \label{indhyp4}
\end{eqnarray}

\smallskip

Then similar (but stronger) estimates hold at the larger scale $N$: 
\begin{equation}\label{indc3}
 \mbox{mes } [ \x \in \mathbb{T}^2 : \abs{ \frac{1}{N} \log \norm{ M_{N} (\x, \lambda, E ) } - L_{N} (\lambda, E ) }  >  S( \lambda )  N^{-\tau} ]  <   e^{-N^\sigma}
\end{equation}
 \begin{eqnarray}
L_{N} (\lambda, E ) & \geq & \gamma S(\lambda)  \label{indc1} \\
& & - 2 [  L_{ N_0} (\lambda, E ) -  L_{2 N_0} (\lambda, E ) ] - C_{0} S (\lambda) N_{0} N^{-1} \notag \\
L_{ N} (\lambda, E ) -  L_{2 N} (\lambda, E ) & \leq & C_{0} S (\lambda) N_{0} N^{-1} \label{indc2}
\end{eqnarray}
for some positive absolute constants $C_0, \tau, \sigma$.

\end{lemma}

\begin{proof}
The parameters $ \lambda$, $E$ and the transformation $\shift = \skews$ or $\shift = \mshift$ are fixed, so they can be suppressed from notations. For instance $ M_N (\x) = M_N (\x, \lambda, E) $, $ S(\lambda) = S $ etc.
 
\smallskip
 
$\blob$ We can assume without loss of generality that $N$ is a multiple of $N_0$, that is, that $ N = n \cdot N_0 $.
 Indeed, if $ N = n \cdot N_0 + r $, $0 \leq r <  N_0 $, then 
\begin{equation}\label{N/N_0}
 \abs{ \frac{1}{N} \log \norm{ M_{N} (\x) } - \frac{1}{n \cdot N_0} \log \norm{ M_{n \cdot N_0} (\x) } } 
 \leq 2 S N_0 N^{- 1}    
\end{equation}
 
Therefore, if we prove (\ref{indc1}), (\ref{indc2}), (\ref{indc3}) at scale $ n \cdot N_0 $, then they hold 
at scale $N$ too.
 
To prove (\ref{N/N_0}), first note that $ M_N (\x) = B (\x) \cdot M_{n \cdot N_0} (\x) $, where 
$$ B (\x) := \prod_{j=N}^{n \cdot N_0 + 1} A (\shift^j \x) = \prod_{j=n \cdot N_0 + r}^{n \cdot N_0 + 1} A (\shift^j \x)$$
so
$$ \norm{ B (\x) }   \leq e^{r \cdot S} \leq e^{N_0 \cdot S} \mbox{  and   } \, \norm{ B (\x) ^{-1} }    \leq e^{r \cdot S} \leq e^{N_0 \cdot S} $$

Since $ \norm{ M_{n \cdot N_0} (\x) } \geq 1  $  and  $ \norm{ M_{N} (\x) } \geq 1  $, it follows that:
 \begin{eqnarray*}
  \frac{1}{N} \log \norm{ M_{N} (\x) } - \frac{1}{n \cdot N_0} \log \norm{ M_{n \cdot N_0} (\x) }    = 
 \frac{1}{n \cdot N_0} \log \frac{\norm{ M_{N} (\x) }^{\frac{n \cdot N_0}{N}}}{\norm{ M_{n \cdot N_0} (\x) }} \\
 \ \leq   \frac{1}{n \cdot N_0} \log \frac{\norm{ B (\x) }^{\frac{n \cdot N_0}{N}} \cdot \norm{ M_{n \cdot N_0} (\x) }^{\frac{n \cdot N_0}{N}}}{\norm{ M_{n \cdot N_0} (\x) }}   \\
\  \leq  \frac{1}{n \cdot N_0} \log \, (e^{N_0 S})^{\frac{n \cdot N_0}{N}} = S N_0 N^{- 1}  
 \end{eqnarray*}

Similarly
 \begin{align*}
 \frac{1}{n \cdot N_0} \log \norm{ M_{n \cdot N_0} (\x) } -  \frac{1}{N} \log \norm{ M_{N} (\x) } =
\frac{1}{n \cdot N_0} \log \frac{|| M_{n \cdot N_0} (\x)||}{ \norm{ M_{N} (\x) }^{\frac{n \cdot N_0}{N}}} 
\end{align*}
\begin{align*}
 = & \ \frac{1}{n \cdot N_0} \, \log \; [ \, \Bigl( \frac{\norm{ M_{n \cdot N_0} (\x) }}{|| M_N (\x) ||} \Bigr)^{\frac{n \cdot N_0}{N}} \cdot \, \norm{ M_{n \cdot N_0} (\x) }^{\frac{r}{N}} \, ] \\
 \leq & \ \frac{1}{n \cdot N_0} \,  \log \; [ \,
\norm{ ( B (\x) )^{-1} }^{\frac{n \cdot N_0}{N}} \cdot \norm{ M_{n \cdot N_0} (\x) }^{\frac{r}{N}} \, ] \\
 \leq & \ \frac{1}{n \cdot N_0} \, \log \; [ \, (e^{ N_0 S})^{\frac{n \cdot N_0}{N}} \cdot (e^{n  N_0 S})^{\frac{N_0}{N}} ] \, = 2 S N_0 N^{- 1}
 \end{align*}
and inequality (\ref{N/N_0}) now follows.

\medskip

$\blob$ We are going to show that (\ref{scale1}) -  (\ref{indhyp4})  allow us to apply the avalanche principle to the ``blocks'' $M_{N_0} (\shift^{(j-1) N_0} \, \x)$, for $j = \overline{1, n}$.  Each of these blocks is a product of $N_0$ matrices, and  they multiply up to  $M_N (\x)$.

\smallskip

Denote the set in (\ref{indhyp1}) by $B_{N_0}$ and similarly the set in (\ref{indhyp2}) by $B_{2 N_0}$.\\
If $ \x \notin B_{N_0} $ then using (\ref{indhyp1}), (\ref{indhyp3}) and \eqref{scale1} we get 
$$ \norm{ M_{N_0} (\x) } >  e^{- \frac{\gamma}{10} S  N_{0} +  L_{N_0} \cdot \, N_0 } \geq e^{\frac{9 \gamma}{10} S  N_{0} } =: \mu > e^{N_0}  \geq N > n$$
so
\begin{equation}\label{avalp1}
 \norm{ M_{N_0} (\x) } \geq \mu \geq n \hspace{.2in} \mbox{ if } \x \notin B_{N_0}
\end{equation} 

For $ 1 \leq j \leq n = \frac{N}{N_0} $ consider $ A_{j} = A_{j} (\x) := M_{N_0} (\shift^{(j-1) N_0} \x) $. Then (\ref{avalp1}) implies
\begin{equation}\label{avalpp1}
\min_{1 \leq j \leq n} \norm{ A_{j} (\x) }\geq \mu  \hspace{.2in} \mbox{ for all } 
 x \notin \bigcup _{j=0}^{n} \shift^{ - j N_0}  B_{N_0}
\end{equation}

Since $ A_{j+1} (\x) \cdot A_{j} (\x) =  M_{2 N_0} (\shift^{(j-1) N_0} \x) $, using (\ref{indhyp1}), (\ref{indhyp2}), (\ref{indhyp4}), for $ \x \notin  \bigcup _{j=0}^{n} (\shift^{ - j N_0}  B_{N_0}) \cup   \bigcup _{j=0}^{n} (\shift^{ - j N_0}  B_{2 N_0}) $ (which is a set of measure \\ $ < 2 N^{- D} \cdot N = 2 N^{- D + 1} $), we have :
\begin{eqnarray*}
 \log \norm{ A_{j+1} (\x) } +  \log \norm{ A_{j} (\x) }- \log \norm{ A_{j+1}(\x) \cdot A_{j}(\x) }\\
  = \log \norm{ M_{N_0} (\shift^{j N_0} \x) } + \log \norm{ M_{N_0} (\shift^{(j-1) N_0} \x) } - \log \norm{ M_{2 N_0} (\shift^{(j-1) N_0} \x) }  \\
  \leq  N_{0} (  L_{N_0} + \frac{S \gamma}{10} ) +  N_{0} (  L_{N_0} + \frac{S \gamma}{10} ) + 2 N_{0} (\frac{S \gamma}{10} -  L_{2 N_0}) \\
 =  2 N_{0} (  L_{N_0} -   L_{2 N_0}) +  \frac{4 S \gamma}{10} N_{0} \leq  \frac{9 S \gamma}{20} N_{0} = \frac{1}{2} \log \mu
 \end{eqnarray*}
Therefore,
\begin{equation}\label{avalpp2}
 \log \norm{ A_{j+1}(\x) } +  \log \norm{ A_{j}(\x) } - \log \norm{ A_{j+1}(\x) \cdot A_{j}(\x) } \leq  \frac{1}{2} \log \mu 
\end{equation}
for $\x$ outside a set of measure $ < 2 N^{- D + 1} $.

Estimates \eqref{avalpp1}, \eqref{avalpp2} are exactly the assumptions in the avalanche principle (Proposition 2.2 in \cite{GS1}). We then conclude:
 \begin{equation}\label{avalpp3}
 \abs{ \log \norm{ A_{n}(\x) \cdot \dotsc \cdot A_{1}(\x) } + \sum_{j=2}^{n-1} \log \norm{ A_{j}(\x) }  - \sum_{j=1}^{n-1} \log \norm{ A_{j+1}(\x) \cdot A_{j}(\x) } } \lesssim \frac{n}{\mu} 
 \end{equation}
for $\x$ outside a set of measure $ < 2 N^{- D + 1} $.

Hence, since $N = n \cdot N_0$ and $ A_{n}(\x) \cdot \dotsc \cdot A_{1}(\x) =  M_{N} (\x)$, we have:
\begin{align*}
\abs{ \log \norm{ M_{N} (\x) } +  \sum_{j=2}^{n-1} \log \norm{ M_{N_0} (\shift^{(j-1) N_0} \x) } \\
 - \sum_{j=1}^{n-1} \log \norm{ M_{2 N_0} (\shift^{(j-1) N_0} \x)} } \,   \lesssim \frac{n}{\mu}
\end{align*}

Therefore
\begin{eqnarray}
\abs{ \frac{1}{N}  \log \norm{ M_{N} (\x) } +  \frac{1}{n} \sum_{j=2}^{n-1} \frac{1}{N_0} \log \norm{ M_{N_0} (\shift^{(j-1) N_0}  \x) } \notag \\
 - \frac{2}{n} \sum_{j=1}^{n-1} \frac{1}{2 N_0} \log \norm{ M_{2 N_0} (\shift^{(j-1) N_0} \x) } }   \lesssim  \frac{1}{\mu} \label{9}
\end{eqnarray}

$\blob$ We will go from averages of $n$ blocks in (\ref{9}), to averages of $N$ shifts.
In (\ref{9}) replace $\x$ by  $\x, \shift \x, \dotsc \shift^{N_{0}-1} \x$ and then average (i.e. add up all these $N_0$ inequalities  and divide by $N_0$) to get:
\begin{eqnarray}
\abs{ \frac{1}{N_0} \sum_{j=0}^{N_{0}-1} \frac{1}{N} \log \norm{ M_{N} (\shift^{j} \x) }  +  
\frac{1}{N} \sum_{j=0}^{N-1} \frac{1}{N_0} \log \norm{ M_{N_0} (\shift^{j} \x) }\notag \\
-  \frac{2}{N} \sum_{j=0}^{N-1} \frac{1}{2 N_0} \log \norm{ M_{2 N_0} (\shift^{j} \x) } }   \lesssim   \frac{1}{\mu} \label{90}
\end{eqnarray}

The almost invariance property - Lemma (\ref{mshift}) implies: 
\begin{equation}\label{900}
 \abs{ \frac{1}{N}  \log \norm{ M_{N} (\x) } -  \frac{1}{N_0} \sum_{j=0}^{N_{0}-1} \frac{1}{N} \log \norm{ M_{N} (\shift^{j} \x) } }\lesssim \frac{S N_0}{N} 
\end{equation}

From (\ref{90}) and (\ref{900}) we get: 
\begin{eqnarray}
\abs{ \frac{1}{N}  \log \norm{ M_{N} (\x) } + 
\frac{1}{N} \sum_{j=0}^{N-1} \frac{1}{N_0} \log \norm{ M_{N_0} (\shift^{j} \x) } \notag \\
-  \frac{2}{N} \sum_{j=0}^{N-1} \frac{1}{2 N_0} \log \norm{  M_{2 N_0} (\shift^{j} \x) } }   \lesssim   \frac{S N_0}{N} + \frac{1}{\mu} \,  \lesssim S N_0 N^{- 1} \label{10}
\end{eqnarray}
for  $\x \notin  B_1 := \bigcup _{j=0}^{N} (T^{ - j }  B_{N_0}) \cup   \bigcup _{j=0}^{n} (T^{ - j }  B_{2 N_0})$ where $\mbox{mes }[ B_1 ] < 2 N^{- D + 1}$.

Integrating the left hand side of (\ref{10}) in $\x$, we get:
 \begin{align}
 \abs{  L_{N} +  L_{N_0} - 2  L_{2 N_0} } & < C S N_0 N^{-1}  + 4 S \cdot 2 N^{- D + 1}  <  C_{0} S N_0 N^{-1} \label{11}\\
L_{N} +  L_{N_0} - 2  L_{2 N_0} & > -  C_0  S N_0 N^{-1} \notag \\
L_{N} & >  L_{N_0} -  2 (  L_{ N_0} -  L_{2 N_0} ) - C_{0} S N_{0} N^{-1} \notag \\
& >  \gamma S  -  2 (  L_{ N_0} -  L_{2 N_0} ) - C_{0} S N_{0} N^{-1} \notag 
\end{align}
which proves \eqref{indc1}.

Clearly all the arguments above work for $N$ replaced by $2 N $, so we get the analogue of (\ref{11}) :
\begin{equation}\label{111} 
\abs{  L_{2 N} +  L_{N_0} - 2  L_{2 N_0} }  <  C_{0} S N_0 N^{-1}
\end{equation}
From (\ref{11}) and (\ref{111}) we obtain
$$ L_{N} -  L_{2 N} \leq  C_{0} S N_0 N^{-1} $$
which is exactly (\ref{indc2}).

\medskip

$\blob$ To prove the LDT (\ref{indc3}) at scale $N$, we are going to apply  the estimate \eqref{shiftldt} on averages of shifts of pluri-subharmonic functions  to the transfer matrix substitutes $u_{N_0} $ and $u_{2 N_0} $. Their widths of subharmonicity in each variable are $\rho_{N_0},  \rho_{2 N_0} \approx N_0^{- \delta-1}$ and they are uniformly bounded by $S$. 

\smallskip

Using (\ref{aproxu}) which holds at scales $N_0$ and $ 2 N_0$ due to (\ref{scale2}), we can `substitute'  in (\ref{10})
$ \frac{1}{N_0} \log \norm{ M_{N_0} (\shift^{j} (\x) } $ by $  u_{N_0} (\shift^j \x) $ and  $  \frac{1}{2 N_0} \log \norm{ M_{2 N_0} (\shift^{j} (\x) } $ by $ u_{2 N_0} (\shift^j \x) $  and get, for $ \x \notin B_1$: 
\begin{equation}\label{12}
\abs{ \frac{1}{N}  \log \norm{ M_{N} (\x) } +  \frac{1}{N} \sum_{j=0}^{N-1} u_{N_0}(\shift^{j} \x) -  \frac{2}{N} \sum_{j=0}^{N-1}  u_{2 N_0}(\shift^{j} \x) } \lesssim S N_0 N^{-1}
\end{equation} 

Applying (\ref{shiftldt}) to $ u_{N_0} $ and $ u_{2 N_0}$ we get :
\begin{equation}\label{13}
\mbox{ mes } [ \x \in \mathbb{T}^2 :  \abs{ \frac{1}{N} \sum_{j=0}^{N-1} u_{N_0}(\shift^{j} \x) - \avg{ u_{N_0} } } >  
S \cdot N_{0}^{\delta+1} \cdot N^{- \tau_0 } ] < e^{- N^{\sigma_0}}
\end{equation}
\begin{equation}\label{13'}
\mbox{ mes } [ \x \in \mathbb{T}^2 :  \abs{ \frac{1}{N} \sum_{j=0}^{N-1} u_{2 N_0}(\shift^{j} \x) - \avg{ u_{2 N_0} } }  >  
S \cdot N_{0}^{\delta+1} \cdot N^{- \tau_0 } ] < e^{- N^{\sigma_0}}
\end{equation}

Denote the union of the two sets in (\ref{13}), (\ref{13'}) by $B_2$. 

Since $N$ satisfies (\ref{scale1}), $$ S \cdot N_{0}^{\delta+1} \cdot N^{-\tau_0} 
< S \cdot (N^{1/A})^{\delta+1} \cdot N^{-\tau_0}  < S \cdot   N^{-\tau_1} \ \text{ where  } \tau_1< \frac{\tau_0}{2}$$

so from (\ref{12}), (\ref{13}), (\ref{13'}) we get:
\begin{eqnarray}
\abs{ \frac{1}{N}  \log \norm{ M_{N} (\x) }  \, + \,  \avg{ u_{N_0} }  \, - \,  2  \avg{ u_{2 N_0} } } \notag \\  
\lesssim S N_0 N^{-1}  +  S  \cdot   N^{-\tau_1}  \lesssim  S  \cdot N^{-\tau_1} \label{14}
\end{eqnarray}

for $\x  \notin B := B_1  \cup  B_2 $, where  $$\mbox{mes } [ B ] < 2 N^{- D + 1} + 2 e^{-N^{\sigma}} <  3 N^{- D + 1} <   N^{- D + 2} $$

Using (\ref{aprox<u>}) at scales $N_0$, $2 N_0$ and taking into account (\ref{scale1}), estimate (\ref{14}) becomes:
\begin{equation}\label{15}
\abs{ \frac{1}{N}  \log \norm{ M_{N} (\x) } \, + \,  L_{N_0}  \, - \,   2 L_{2 N_0} }  < 2 S \cdot N^{-\tau_1} \, + \, 2 e^{-  N_{0}^2} < 3 S N^{-\tau_1}  
\end{equation}
provided $x \notin B$.

\smallskip

Combine (\ref{15}) with (\ref{11}) to get:
\begin{equation}\label{16}
\abs{ \frac{1}{N}  \log \norm{ M_{N} (\x) } \, - \,  L_{N} }  <  C_{0} S N_0 N^{-1} + 3 S \cdot N^{-\tau_1}  <  S \cdot N^{-\tau_2} 
 \end{equation}
for all  $ \x \notin B $, where $\mbox{mes }[ B ] <  N^{- D + 2} $ and $\tau_2 <  \tau_1$.

\medskip

However, (\ref{16}) is not exactly what we need in order to prove the estimate (\ref{indc3}). We have to prove an estimate like (\ref{16}) for $\x$ outside an exponentially small set, and we only have it outside a polynomially small set. 
To boost this estimate, we employ again Lemma~\ref{boost}.

From (\ref{16}), using again (\ref{aproxu}), (\ref{aprox<u>}) at scale $N$, we get:
 \begin{equation}\label{160}
 \mbox{ mes } [ \x \in \mathbb{T}^2 :  \abs{ u_{N} (\x) - \avg{u_N} } > S \cdot N^{-\tau_2} ] < N^{- D + 2}
 \end{equation}
 
We apply Lemma~\ref{boost} to $ u (\x) := \frac{1}{S} u_N (\x)$, which is a pluri-subharmonic function on the strip $\Aa_{\ro_N}$, with upper bound $B = 1$ on this strip. 
 
Estimate (\ref{160}) implies   
 \begin{equation}\label{1600}
 \mbox{ mes } [ \x \in \mathbb{T}^2 :  \abs{ u (\x) - \avg{u} } >  N^{-\tau_2} ] < N^{- D + 2}
 \end{equation}
 
Then for  $ \epsilon _0 := N^{-\tau_2} $, $ \epsilon _1 := N^{- D + 2}$, $B = 1$, $\rho = \rho_N \approx N^{-\delta-1} $ we have 
\begin{align*}
{\epsilon _0}^{1/4} + \sqrt{\frac{B}{\rho}}  \;   \frac{{\epsilon_1}^{1/4}} {{\epsilon_0}^{1/2}}   = 
 N^{-{\tau_2}/4} + N^{\frac{\delta+1}{2}} \, N^{- \frac{D+2}{4}} \, N^{{\tau_2}/2} \\
 = N^{-{\tau_2}/4}  + N^{-1} \,  N^{{\tau_2}/2} <  N^{-\sigma_1}
 \end{align*}
 for some positive constant  $\sigma_1$.

The conclusion \eqref{strong} of Lemma~\ref{boost} then boosts (\ref{160}) from a small deviation outside a polynomially small set, to one outside an exponentially small set, amid a small power loss in the deviation:
 \begin{equation}
 \mbox{ mes } [ \x \in \mathbb{T}^2 :   \abs{ u_N (\x) - \avg{u_N} } >  S \, N^{-{\tau_2}/4} ] < e^{-c N^{\sigma_1}} < e^{-N^{\sigma}}
 \end{equation}
which proves estimate (\ref{indc3}).
\end{proof}

\begin{remark}{\rm The scaling factor $\sqrt{\frac{B}{\rho}}$ in estimate \eqref{strong} of Lemma~\ref{boost} is what prevents this approach via polynomial approximation to extend to more general Carleman classes of potential functions. This is because when the estimates on the Fourier coefficients of the potential function are weaker than estimate \eqref{fcoef} for Gevrey functions, the size $\rho = \rho_N$ of the holomorphic extension of the $N$th transfer matrix substitute will cancel any decay in the expression   $\sqrt{\frac{B}{\rho}}  \;   \frac{{\epsilon_1}^{1/4}} {{\epsilon_0}^{1/2}}$
}
\end{remark}

\medskip


We will combine the base step (Lemma~\ref{base-step}) with the inductive step (Lemma~\ref{ind}) to prove the large deviation estimate for transfer matrices and the positivity of the Lyapunov exponent. The proof of the LDT will also provide us with the major ingredient for deriving the continuity of the Lyapunov exponent.

\smallskip

\begin{theorem}\label{LDT}
Consider the Schr\"{o}dinger operator (\ref{op1}) on $l^2(\Z)$:
$$
[H (\x) \, \psi]_n := - \psi_{n+1} - \psi_{n-1} + \la \, v (\shift^n \x) \, \psi_n
$$
where the transformation $\shift$ is either the skew-shift \eqref{skew} or the multi-frequency shift \eqref{multishift}.
Assume that for some $\ka > 0$ the underlying frequency satisfies the Diophantine condition $DC_\ka$ described in \eqref{DC} or \eqref{DCM} respectively.

Assume moreover that the potential function $v (\x)$ belongs to a Gevrey class $G^s (\T^2)$ and that it is transversal as in \eqref{TC}. 

Then there exists $ \lambda_{0} = \lambda_{0} ( v, \kappa ) $ so that for every fixed $\lambda $ with $\abs{\la} \geq \lambda_{0} $ and for every energy $E$, we have:
\begin{equation}\label{ldt}
 \mbox{mes } [ \x \in \mathbb{T}^2 : \abs{ \frac{1}{N} \log \norm{ M_{N} (\x, \lambda, E ) } - L_{N} (\lambda, E ) } >  N^{-\tau} ]  <   e^{- {N^{\sigma}}}
\end{equation}
for some absolute constants $ \tau, \sigma  > 0 $, and for all
$ N \geq  N_{0} (\lambda, \kappa, v, s) $.
   
Furthermore, for every such transformation $\shift$ and coupling constant $\lambda$ and for all energies $E \in \mathbb{R}$ we have:
\begin{equation}\label{poslyap}
 L (\la, E) \geq \frac{1}{4} \log\abs{\la} > 0
\end{equation}
\end{theorem}

\begin{proof} We refer to the list of constants preceding Lemma \ref{ind}. 

We use the initial step - Lemma \ref{base-step} at a sufficiently large initial scale $N_0 \ge N_{0 0} = N_{0 0} (v)$. We will explain how the scale $N_0$ is chosen later. We get  constants $ \lambda_1 $, $B$ $ > 0$ such that for every $\lambda$ with 
$\sabs{\la}  \ge \max \{ \lambda _1, (2 N_0) {^B} \} $ (we want Lemma \ref{base-step} to apply at both scales $N_0$ and $2 N_0$) we have:
\begin{align}
  \mbox{mes } [ \x \in \mathbb{T}^2 : \abs{ \frac{1}{N_0} \log \norm{ M_{N_0} (\x) } - 
L_{N_0}  } >  \frac{1}{20} S  ]  <   N_0^{- A^2 \cdot D} \leq  N^{- D} \label{ldtN0} \\
 \mbox{mes } [\x \in \mathbb{T}^2 : | \frac{1}{2 N_0} \log || M_{2 N_0} (\x) || -L_{2 N_0}  | >  \frac{1}{20} S  ]  \notag \\ 
 <   (2 N_0)^{- A^2 \cdot D} \less N^{- D}  \label{ldt2N0} 
 \end{align}
\begin{align}
 L_{ N_0}, \, L_{2 N_0} \geq \frac{1}{2} S \label{LN0>} \\
L_{ N_0} -  L_{2 N_0} \leq  \frac{1}{80} \, S \label{LN0-L2N0}
\end{align}

Of course \eqref{ldtN0} and \eqref{ldt2N0} hold provided $N$ satisfies:
\begin{equation}\label{N<}
N  \leq N_0^{A^2} 
\end{equation}

Estimates \eqref{ldtN0} - \eqref{LN0-L2N0} above are exactly the assumptions \eqref{indhyp1} - \eqref{indhyp4} (at scale $N_0$, with $ \gamma = \gamma_0 = \frac{1}{2}$) in Lemma \ref{ind}. of the inductive 
step of LDT. 

\medskip

However, in order to apply this inductive step lemma and obtain similar estimates at the larger scale $N$, the initial scale $N_0$ and the disorder $\la$ have to satisfy the condition \eqref{scale2}. Together with the conditions on $\la$ and $N_0$ from the initial step (Lemma \ref{base-step}),  $N_0$ and $\la$ have to satisfy:
\begin{align}   
(2 N_0)^B \leq \abs{\la}  \leq e^{N_0} \label{alfa} \\
N_0 \geq N_{0 0} \label{beta} \\
\abs{\la}  \geq  \lambda_1\label{gama} 
\end{align}

We want to prove the LDT for every disorder $  \lambda $ large enough, 
$ \abs{\la}  \geq \lambda_0 $ and not just for $\lambda$ in a bounded interval as in \eqref{alfa}. To do that, we will have to first choose $\la$ large enough, and then to pick $N_0 = N_0 (\la) \ge N_{0 0}$ appropriately. Here is how we can accomplish that.

 The condition \eqref{alfa} is equivalent to 
 \begin{equation}\label{alfa'}
\log \sabs{\la} \leq N_0 \leq  \frac{1}{2}\abs{\la}^{1/B}
\end{equation}

We can find $ \lambda_0 $ large enough, $ \lambda_0 =  \lambda_0 (v, \kappa)$, $\la_0 \ge \la_1$, so that if $\sabs{\la} \geq   \lambda_0 $, then 
\begin{equation}\label{betagama'}
 \log \sabs{\la} \geq   N_{0 0} \, \mbox{  and  } \,  
 \log \sabs{\la} \ll   \frac{1}{2}\abs{\la}^{1/B}
 \end{equation}
 
Then for every such $\lambda$ we can pick $  N_0 = N_0 ( \lambda)$ so that \eqref{alfa'} holds. Combining this with \eqref{betagama'}, we get that  (\ref{alfa}), (\ref{beta}), (\ref{gama}) hold.

\medskip

All the assumptions on the small scale $N_0$ in the inductive step - Lemma \ref{ind} hold now, so if we choose the large scale $N$ such that
\begin{equation}\label{N}
 N_0^{A} \leq N \leq  N_0^{A^2} ( <  e^{N_0}) 
\end{equation}
then \eqref{N<} and \eqref{scale1} hold, so we can apply Lemma \ref{ind} to get:
\begin{equation}\label{fin7}
 \mbox{mes } [ \x \in \mathbb{T}^2 : \abs{ \frac{1}{N} \log \norm{ M_{N} (\x) } - L_{N}  } >  S N^{-\tau} ]  <   e^{-N^\sigma}
\end{equation}
\begin{align}
L_{N}  & \ \geq \gamma_0 S - 2 (  L_{ N_0}  -  L_{2 N_0} ) - C_{0} S N_{0} N^{-1} \label{fin5}\\
L_{ N}  -  L_{2 N}  & \ \leq  C_{0} S  N_{0} N^{-1} \label{fin6}
\end{align}
for some positive absolute constants $C_0, \tau, \sigma$.

 \smallskip
 
Estimate (\ref{fin7}) proves the LDT (\ref{ldt}) at scale $N$ in the range  
 $ [ N_0^{A}, N_0^{A^2} ] $.
If $N_1$ is in this range, say $N_1 =  N_0^{A}$, then  (\ref{fin6}) and (\ref{fin5}) imply:
$$ L_{ N_1}  -  L_{2 N_1}  \leq  C_{0} S  N_{0} N^{-1} $$
$$  L_{N_1}   \geq \gamma_0 \, S - 3 \, C_{0} S  N_{0} N^{-1} =  \gamma_0 \, S - 3  \, C_{0}  N_{0}^{- A + 1} S  =: \gamma_1 \cdot S $$
where $$\gamma_1 := \gamma_0 -  3  \, C_{0}  N_{0}^{- A + 1} =  \frac{1}{2} - 3  \, C_{0}  N_{0}^{- A + 1} >  \frac{1}{4} $$
provided we chose $N_{0 0}$ (and so $N_0$) large enough depending on $A$, $C_0$.

\smallskip

Therefore we have:
\begin{equation}\label{fin8}
L_{N_1}  \geq \gamma_1 S
\end{equation}
and
$$L_{ N_1}  -  L_{2 N_1}   \leq  C_{0} S  N_{0} N^{-1}  = C_{0}  S N_{0}^{- A + 1}  < \frac{1}{160} \cdot S 
<  \frac{\gamma_1}{40} \cdot S$$
so
\begin{equation}\label{fin9}
L_{ N_1}  -  L_{2 N_1}    <  \frac{\gamma_1}{40} \cdot S
\end{equation}

Since $2 N_1 = 2 N_0^{A} $ is in the range  $ [ N_0^{A}, N_0^{A^2} ] $,  (\ref{fin8}) holds at scale $2 N_1$ too, so we have:
\begin{equation}\label{fin10}
L_{N_1}, L_{2 N_1}  \geq \gamma_1 S
\end{equation}

Choosing the next large scale  $N_2$ so that  $ N_1^{A} \leq N_2  \leq  N_1^{A^2} (< e^{N_1})$, we have $  e^{- {N_{1}^{\sigma}}} < N_{1}^{- A^2 \cdot D} \leq N_{2}^{- D}$, so  (\ref{fin7}) implies: 
\begin{align}
 \mbox{mes } [  \x \in \mathbb{T}^2 \colon \abs{ \frac{1}{N_1} \log \norm{ M_{N_1} (\x) } - 
L_{N_1} } >  \frac{1}{20} S  ]  <    e^{- {N_{1}^{\sigma}}} < N_{2}^{- D}    \label{N1} \\
 \mbox{mes } [ \x \in \mathbb{T}^2 : \abs{ \frac{1}{2 N_1} \log \norm{ M_{2 N_1} (\x) } -L_{2 N_1}  } >  \frac{1}{20} S  ]  \less   N_{2}^{- D} 
 \label{N2}
 \end{align}

Estimates \eqref{N1}, \eqref{N2}, \eqref{fin10}, \eqref{fin9} are the assumptions in the inductive step - Lemma \ref{ind} with small scale $N_1$ and large scale $N_2$, where  $N_2 \in  [ N_1^{A}, N_1^{A^2} ] \, = \,  [ N_0^{A^2}, N_0^{A^3} ]$.  Applying Lemma \ref{ind}, we  get the LDT \eqref{ldt} for $N_2$ in this range. Moreover, we get: 

$$  L_{ N_2}  -  L_{2 N_2}  \leq  C_{0} S  N_{1} N_{2}^{-1} $$
and
$$ L_{N_2}  \geq \gamma_1 S - 2 (  L_{ N_1}  -  L_{2 N_1} ) -
 C_{0} S N_{1} N_{2}^{-1} \geq (\gamma_1 - 3 C_0  N_{1}^{- A + 1}) \cdot S
=: \gamma_2 \cdot S $$
where  $$\gamma_2 := \gamma_1 -  3  C_{0}  N_{1}^{- A + 1} \ge
 \frac{1}{2} -  3  C_{0}  N_{0}^{- A + 1} -  3  C_{0}  N_{0}^{A \cdot (- A + 1)}
> \frac{1}{4} $$
again, provided  $N_{0 0}$ (thus $N_0$) was chosen large enough depending on $A$, $C_0$.

Hence we have
$ L_{N_2} \geq \gamma_2 \cdot S $ and
$  L_{N_2} -  L_{2 N_2} \leq \frac{\gamma_2}{40} \cdot S $.

Continuing this inductively, we obtain (\ref{ldt}) at every scale $ N \geq N_{0}^A $.

Also, at each step $k$ in the induction process, if $ N \in  [ N_{k}^{A}, N_{k}^{A^2} ] $, then
$L_N \geq \gamma_k \cdot S > \frac{1}{4} \cdot S  $ so
$$ L = \inf_{N} L_N \geq \frac{1}{4} \cdot S $$ 
and (\ref{poslyap}) is proven.
\end{proof}

\smallskip

We now prove that the Lyapunov exponent is continuous as a function of the energy.


\begin{theorem}\label{cont}
Under the same conditions as in Theorem \ref{LDT} above, and for any $\sabs{ \la } \ge \la_0 (v, \ka)$, the Lyapunov exponent $L (E)$ is a continuous function of the energy $E$ with modulus of continuity on each compact interval $\mathcal{E}$ at least:
\begin{equation}\label{modcont-weak}
w (t) = C \, \bigr(\log \frac{1}{t} \bigl)^{-\beta}
\end{equation} 
where $C = C (\mathcal{E}, \la, v, \ka, s)$ and $\beta \in (0,1)$ is a universal constant that can be chosen, at the expense of $C$, to be arbitrarily close to 1.  

\end{theorem}

\begin{proof} We will fix $\la, \shift$ and omit them from notations. We also fix the compact interval $\mathcal{E}$.

It is easy to show (see below) that for every scale $N$, the functions $L_N (E) $ are (Lipschitz) continuous. To prove that their limits $L (E)$ are also continuous with a certain modulus of continuity, we need a quantitative description of the convergence $L_N (E) \to L (E) $ as $N \to \infty$. The better this rate of convergence, the sharper the modulus of continuity of $L (E)$. 

It follows from the proof of Theorem \ref{LDT} above (see \eqref{fin6} and the inductive process thereafter) that for every scales $N_0$ and $N$ such that  $N_0 \ge N_{0 0} (\la, v, \ka)$ and $N_0^A \le N \le N_0^{A^2}$, we have:
$$L_N (E)  - L_{2N}  (E) \less N_0 N^{-1} \le N^{1/A} \, N^{-1} =: N^{- \beta}$$
so
\begin{equation}\label{cont1}
L_N (E) - L_{2N} (E)  \less N^{- \beta} \ \text{ for all } N \ge N_{0 0}
\end{equation}
Summing up over dyadic $N$'s we conclude:
\begin{equation}\label{cont2}
L_N (E) - L (E)  \less N^{- \beta} \ \text{ for all } N \ge N_{0 0}
\end{equation}
which is the quantitative convergence we were seeking.  

\smallskip

To show that $$L_N (E) = \frac{1}{N} \, \int_{\T^2} \log \norm{ M_N (\x, E) } \, d \x$$
are continuous, we use Trotter's formula for the transfer matrix $M_N (\x, E)$:

$$ M_{N} (\x, E) -  M_{N} (\x, E^\prime)  = $$
$$=  \sum_{j=1}^{N} A (\shift^{N} \x, E) \ldots  \,  [A (\shift^j \x, E) - A (\shift^j \x, E^\prime)] \,  \ldots A (\shift  \, \x, E^\prime)$$
But
$$ A (\shift^j \x, E) - A (\shift^j \x, E^\prime) = \Bigl[\begin{array}{cc}
E^\prime - E  &   0  \\
0 &  0 \\  \end{array} \Bigr]$$ 
and 
$$  \norm{ A (\shift^j \x, E) }  \leq  e^{S} \quad \mbox{ for all } E \in \mathcal{E} $$
so
$$ \norm{ M_{N} (\x, E) -  M_{N} (\x, E^\prime) } \le e^{S N} \, \abs{ E - E^\prime } $$

Therefore, since $ \norm{ M_{N} (\x, E) } \geq 1 $ and  $ || M_{N} (\x, E^\prime) || \geq 1$, we have: 
\begin{align*}
\abs{ \log \norm{ M_{N} (\x, E) } -  \log \norm{ M_{N} (\x, E^\prime) } }  \\
\le \norm{ M_{N} (x, E) -  M_{N} (x, E^\prime) } \le e^{S N} \, \abs{ E - E^\prime }
\end{align*}

Integrating in $\x$ we obtain:
\begin{equation}\label{cont3}
| L_{N} (E) - L_{N} (E^\prime) | \leq e^{S N} \, | E - E^\prime | 
\end{equation}
which shows Lipschitz continuity for the maps $L_N (E)$. 

Combining \eqref{cont2} and \eqref{cont3} we obtain:
\begin{equation}\label{cont4}
\abs{ L (E) - L (E^\prime) } \less N^{- \beta} + e^{S N} \, | E - E^\prime | \quad \text{ for all } N \ge  N_{0 0} (\la, v, \ka)
\end{equation} 

For every such $N$ let $$\abs{ E - E^\prime } \sim e^{- S N} \, N^{- \beta}$$ so
$$\abs{ L (E) - L (E^\prime) } \less N^{- \beta} $$
Since
$$\log \frac{1}{\abs{ E - E^\prime }} \sim S N + \beta \log N \less S N$$
we have
$$N^{- \beta} \sim \bigl(\frac{1}{S}\bigr)^{- \beta} \, \bigl( \log \frac{1}{\abs{ E - E^\prime }} \bigr)^{-\beta} = C \,  \bigl( \log \frac{1}{\abs{ E - E^\prime }} \bigr)^{-\beta}$$
where $ C = C (\la, v, \ka)$.

We conclude, using the compactness of $\mathcal{E}$, that for some constant $C = C(\mathcal{E}, \la, v, \ka)$, and for a constant $\beta$ that can be chosen arbitrarily close to $1$ by starting off with a large enough constant $A$, we have:
$$\abs{ L (E) - L (E^\prime) }  < C \,  \bigl( \log \frac{1}{\abs{ E - E^\prime }} \bigr)^{-\beta}$$
\end{proof}

\begin{remark}\rm{The rate of convergence \eqref{cont2} can be improved to 
\begin{equation}\label{cont2'}
\abs{ L(E) + L_N (E) - 2 L_{2 N} (E) }  \less e^{-c N^{ \eta}} \ \text{ for all } N \ge N_{0 0}
\end{equation}
which follows from the proof of the inductive step, Lemma \ref{ind} (see estimate \eqref{11}) and uses the avalanche principle.
This faster rate of convergence leads to the sharper modulus of continuity \eqref{modcont} (see \cite{sK1}, \cite{BGS} for details). 
}
\end{remark}


\smallskip

We will now explain how Anderson localization is derived from the large deviation theorem~\ref{LDT}.

Given the Schr\"{o}dinger operator 
\begin{equation}\label{3op}
[H (\x) \, \psi]_n := - \psi_{n+1} - \psi_{n-1} + \la \, v (\shift^n \x) \, \psi_n
\end{equation} 
for every scale $N$ we denote 
$$H_N (\x) := R_{[1, N]} H (\x) R_{[1, N]}$$
where $R_{[1, N]}$ is the coordinate restriction to $[1, N] \subset \field{Z}$ with Dirichlet boundary conditions.

Then the associated Green's functions are defined as
$$G_N (\x, E) := [ H_N (\x) - E ]^{-1}$$
if the $N \times N$ matrix $ H_N (\x) - E $ is invertible.

The large deviation estimate \eqref{ldt} implies, via Cramer's rule,  `good bounds' on the Green's functions $G_N (\x, E)$ associated with (\ref{3op}).

Indeed, for $1 \leq n_1 \leq n_2 \leq N$, we have:
\begin{align}
G_N (\x, E) (n_1, n_2) = [ H_N (\x) - E ]^{-1} (n_1, n_2) \\
= \frac{ \mbox{ det } [ H_{n_1-1} (\x) - E ]  \cdot   \mbox{ det } [ H_{N-n_2} (\shift^{n_2} \x) - E ] }
{\mbox{ det } [ H_{N} (\x) - E ] }
\end{align}

There is the following relation between transfer matrices and determinants:
\begin{equation}\label{tmdet}
M_{N} (\x, E) = \Bigl[ \begin{array}{ccc}
\mbox{ det } [ H_N (\x) - E ]   & & - \mbox{ det } [ H_{N-1} (\shift \x) - E ]  \\
 \mbox{ det } [ H_{N-1} (\x) - E ] & &  - \mbox{ det } [ H_{N-2} (\shift \x) - E ] \\  \end{array} \Bigr]  
\end{equation}

Therefore, we get the following estimate on the Green's functions:
$$\abs{ G_N (\x, E) (n_1, n_2) } \le \frac{\norm{ M_{n_1} (\x, E) }  \cdot \norm{ M_{N-n_2} (\shift^{n_2} \x, E) }}
{\abs{ \mbox{ det } (H_{N} (\x) - E) }}$$

Combining this with the LDT \eqref{ldt}, we obtain the following bounds
on the Green's functions $G_{\Lambda} (E, \x)$ associated with the operator \eqref{3op}. 

For every $N$ large enough and for every energy $E$, there is a set 
$\Omega_{N} (E) \subset \mathbb{T}^2 $ with  $ \mbox{mes } [ \Omega_{N} (E) ] < e^{- N^{\sigma}} $
so that for any $ \x \notin \Omega_{N} (E) $, one of the 
intervals 
$$ \Lambda = \Lambda (\x) = [1, N ], [1, N - 1], [2, N ], [2, N - 1] $$ 
will satisfy : 
\begin{equation}\label{3green} 
 | G_{\Lambda} (E, \x) (n_1 , n_2 ) | < e^{- c  | n_1  - n_2 | + N^{1 -}}
 \end{equation} 

Since $ v(\x) = \sum_{\li \in \mathbb{Z}^2} \hat{v} (\li) e^{2 \pi i \, \li \cdot \x}$
and
$ \abs{ \hat{v} (\li) }  \leq  M e^{- \rho \sabs{ \li }^{1/s}} $ for all $ \li \in \mathbb{Z}^2$, substituting in (\ref{3green}) $v (\x) $ by $  v_1 (\x) := \sum_{|\li| \leq C N^{s}} \hat{v} (\li) e^{2 \pi i \, \li \cdot \x} $
we can assume that the `bad set' $\Omega_{N} (E) $ above 
not only has exponentially small measure, but it also has bounded algebraic 
complexity - it is semi-algebraic of degree  $\, \leq N^{d(s)}.$ 

These sets depend on the energy $E$. The rest of the proof of  localization for \eqref{3op}  involves the elimination of the energy, which uses 
semi-algebraic set theory, and follows exactly the same pattern as the proof 
of the corresponding result for the analytic case (see \cite{BGS}, \cite{BG} or  Chapter 15 in \cite{B}). 

Our statement for the skew-shift model is weaker than the one for the multi-frequency shift, since they both mirror the corresponding results in the analytic case.  

\smallskip

\begin{remark}\rm{
We do not know if the transversality condition \eqref{TC} is indeed necessary, either for the models considered here or for the single-variable shift considered in \cite{sK1}. In particular, we do not know if the Lyapunov exponent is still positive throughout the spectrum for potential functions that have flat parts but are very smooth otherwise. This is a difficult and interesting problem.

Finally, a more challenging problem regarding Gevrey potential functions is proving localization for a long range model, one where the Laplacian is replaced by a Toeplitz matrix. In the case of the skew-shift dynamics, this could lead to applications to more general quantum kicked rotator equations. }
\end{remark}

\bigskip

\nocite{*}
\bibliographystyle{amsplain} 
\bibliography{references}

\end{document}